\documentclass[11pt]{article}

\usepackage{fullpage}
\usepackage{amsmath, amsthm, amssymb}
\usepackage{bbm}
\usepackage{authblk}
\usepackage{mathtools}
\usepackage{graphicx}
\usepackage{xcolor}
\usepackage{enumerate}
\usepackage{booktabs}
\usepackage{natbib}
\usepackage{subfig}
\usepackage{multirow}
\usepackage{float}
\usepackage{tikz}
\usepackage{longtable}
\usepackage{indentfirst}
\usepackage{setspace}
\usepackage{xurl}
\usepackage{enumitem}
\RequirePackage[colorlinks = true, linkcolor = blue, citecolor=blue,urlcolor=blue]{hyperref}

\newtheorem{lemma}{Lemma}

\newtheorem{theorem}{Theorem}

\newtheorem{prop}{Proposition}

\newcommand{\dd}{\,\mathrm{d}\,}
\DeclareMathOperator*{\argmin}{argmin}

\DeclareMathOperator{\MSE}{MSE}

\DeclareMathOperator{\Var}{Var}
\DeclareMathOperator{\Span}{Span}
\DeclareMathOperator{\sign}{sign}
\DeclareMathOperator{\dist}{dist}
\definecolor{c1}{rgb}{0,  0, 0}
\definecolor{c2}{rgb}{0,  0, 0}
\definecolor{c3}{rgb}{128,  0, 128}
\definecolor{ccol}{rgb}{.9,.35,0}

\DeclareRobustCommand\full  {\tikz[baseline=-0.6ex]\draw[c1, line width=0.25mm] (0,0)--(0.4,0);} 

\begin{document}

\title{Penalized Spline M-Estimators for Discretely Sampled Functional Data: Existence and Asymptotics}
\author{Ioannis Kalogridis}
\affil{School of Mathematics and Statistics, University of Glasgow, United Kingdom}
\date{\today}
\maketitle
	
\begin{abstract}

Location estimation is a central problem in functional data analysis. In this paper, we investigate penalized spline estimators of location for discretely sampled functional data under a broad class of convex loss functions. Our framework generalizes and extends previously derived results for non-robust estimators to a broad penalized M-estimation framework. The analysis is built on two general-purpose non-asymptotic theoretical tools: (i) a non-asymptotic existence result for penalized spline estimators under minimal design conditions, and (ii) a localization lemma that captures both stochastic variability and approximation error. Under mild assumptions, we establish optimal convergence rates and identify the small-- and large--knot regimes, along with the critical breakpoint known from penalized spline theory in nonparametric regression. Our results imply that parametric rates are attainable even with discretely sampled data and numerical experiments demonstrate that our estimators match or exceed  the robustness of competing estimators while being considerably more computationally attractive.
\end{abstract}



\section{Introduction}

Functional data analysis (FDA) is a modern branch of Statistics dealing with the analysis of complex high-dimensional objects that have an underlying functional structure \citep[see, e.g.][]{Ramsay:2005}. Consider, in particular, a sample of independent and identically distributed (i.i.d.) $\mathcal{L}^2([0,1])$ processes $X_1, \ldots, X_n$. A fundamental problem in FDA is the estimation of the mean function $\mu_0(t) = \mathbb{E}[X(t)]$ from noisy, discretely observed functional data of the form
\begin{align}
\label{eq:discrfd}
Y_{ij}=X_i(T_{ij})+\zeta_i(T_{ij}), \quad (j=1, \ldots, m_i;\ i=1,\ldots,n),
\end{align}
where  the $\zeta_i$ are independent and identically distributed $\mathcal{L}^2([0,1])$ measurement error or noise processes and $T_{ij} \in [0,1]$ are random sampling points. The mean function $\mu_0(t)$ provides not only a concise summary of the main features of the data, but also a crucial ingredient  in downstream tasks such as functional principal component analysis, regression, and classification.

If the random functions $X_1, \ldots, X_n$ were completely observed without noise, the mean function $\mu_0$ could be readily estimated, for example, with the sample mean $\bar{X}(t) = n^{-1} \sum_{i=1}^n X_i(t)$. Under finite second moments, $\bar{X}$ converges to $\mu_0$ at the parametric rate $n^{-1}$ in the $\mathcal{L}^2([0,1])$-metric \citep[see, e.g.][Section 8.1]{Hsing:2015}, so that in this ideal scenario $\bar{X}$ is an accurate estimator. In practical settings such as \eqref{eq:discrfd}, however, the sample mean performs poorly for two important reasons: (i) there may be regions of $[0,1]$ with few or even no design points making the sample mean highly variable and (ii) the sample mean directly incorporates noise from the $\zeta_i(T_{ij})$. These drawbacks are also shared by other estimators requiring completely observed data, such as the popular spatial median \citep[see, e.g.][]{Gerv:2008}. These issues mean that in most practical situations estimators relying on completely observed noise-free data are not viable and alternative estimation methods need to be considered. 

To address these notable shortcomings of estimators requiring complete data, several smoothing estimators have been proposed in the FDA literature over the years. Examples include the smoothing spline estimator of \citep{Cai:2011}, who also derived the optimal rates of convergence in the setting of discretely observed functional data, the kernel local linear estimator of \citep{Li:2010,Zhang:2016}, and the penalized spline estimator of \citep{Xiao:2020}. These estimators extend nonparametric regression techniques to functional data with repeated (dependent) measurements and, under second-moment conditions, achieve the optimal rates of convergence. A common weakness, nevertheless, is that, since all three of these estimators rely on the minimization of penalized or local square losses, they are extremely sensitive to outlying observations: even a single outlying $Y_{ij}$ can cause severe distortion on the estimates.

Robust functional data analysis with discretely sampled functional data is a far less developed field of study. For location estimation, the only theoretically investigated robust  estimator of which we are aware is the smoothing spline estimator of \citep{Kal:2023b}, who replaced the square loss in \citep{Cai:2011} with a robust loss function (e.g., the absolute loss) and showed that the optimal rates of convergence remain attainable.  While robust, a serious drawback of the approach of these authors is its dimensionality: smoothing spline estimators require as many basis functions as the number of distinct $T_{ij}$ becoming computationally prohibitive unless the design grid is very small, that is, the functional data is sparsely observed. The difficulty is compounded by the absence of closed-form solutions for general loss functions, which forces the use of iterative algorithms that tend to converge slowly in high-dimensional problems.

As a remedy to either the lack of robustness or the computational cost of existing methods, this paper studies a flexible family of lower-rank penalized spline M-estimators. Robustness is achieved via the choice of suitable loss function, while computational efficiency follows from the construction itself: the number of basis functions can be chosen independently of the grid size. We derive detailed convergence rates for these estimators and show that they remain asymptotically optimal, thereby generalizing the results of \citet{Xiao:2020} from the squared loss to a broad family of loss functions. Our novel methodology based on functional analysis may be of independent mathematical interest. We also give an easily checkable finite-sample existence condition for penalized spline M-estimators, applicable far beyond the FDA setting that is primarily considered herein. 

The rest of the paper is organized as follows. Section~\ref{sec:2} describes the proposed estimators and presents our general finite-sample existence theorem. Section~\ref{sec:3} develops the asymptotic theory in a reproducing kernel Hilbert space (RKHS) framework, including an important localization result of independent interest, and establishes optimal asymptotic convergence rates.  Section~\ref{sec:4} reports numerical studies confirming the robustness and computational gains of our method while Section~\ref{sec:5} contains our concluding remarks. For the convenience of the reader, we include our proofs or proof sketches within the main text and defer the supporting technical lemmas to the appendix at the end of this paper.

\section{Robust Penalized Spline Estimators for Functional Data}
\label{sec:2}

To introduce the proposed family of estimators, we recall that a spline is a piecewise polynomial that is smoothly connected at the joints (knots). Specifically, for any integers $p \geq 1$ and $K \geq 1$, let $S_{K}^p([0,1])$ denote the set of spline functions with interior knots $0<x_{p+1}<\ldots<x_{K+p}<1$. Then, for $p=1$, $S_{K}^p([0,1])$ consists of step functions with jumps at the interior knots while for $p \geq 2$:
\begin{eqnarray*}
S_{K}^p([0,1]) = \left\{ s \in \mathcal{C}^{p-2}([0,1]) :\ s(t)\ \text{is a polynomial of degree } (p-1) \text{ on each } [x_i, x_{i+1}] \right\}.
\end{eqnarray*}
It can be shown that $S_{K}^p([0,1])$ is a vector space of dimension $K+p$ and 
a numerically convenient basis is given by the B-splines \citep[see, e.g.][Chapter IX]{DB:2000}. Specifically, define $2p$ additional knots by setting $x_{1}=\ldots =x_p = 0$ and $x_{K+p+1} = \ldots =  x_{K+2p}=1$. The B-splines are defined recursively as
\begin{align*}
B_{j,1}(t) & = \begin{cases} 1, & x_j\leq t < x_{j+1}, \\ 
0, & \text{otherwise}, \end{cases}
\\ 
B_{j,p+1}(t) & = \frac{t-x_{j}}{x_{j+p+1}-x_j}B_{j,p}(t) + \frac{x_{j+p}-t}{x_{j+p}-x_{j+1}}B_{j+1,p}(t),
\end{align*}
with the convention $0/0 = 0$. This gives rise to precisely $K+p$ B-spline functions $B_{1,p}, \ldots, B_{K+p,p}$. It is not difficult to see that each $B_{j,p}: [0,1] \to [0,1]$ and its support $\overline{\{t \in [0,1]: B_{j,p}(t) > 0\}}$ is $[x_j,x_{j+p}]$. This implies that any two B-splines $B_{j,p}(t)$ and $B_{k,p}(t)$ are orthogonal in $\mathcal{L}^2([0,1])$ so long as $|j-k| \geq p$.

We now embed this spline construction into the framework of penalized M-estimators. Let $p \in \mathbbm{N}$ be such that $p \geq 2$ and choose an integer $r \in \mathbbm{N}$ such that $r < p$. The proposed estimator for $\mu_0$ is given by 
\begin{align}
\label{eq:est}
\widehat{\mu}_n = \argmin_{f \in S_{K}^p([0,1])} \left[ \frac{1}{n} \sum_{i=1}^n \frac{1}{m_i} \sum_{j=1}^{m_i} \rho\left(Y_{ij}-f(T_{ij}) \right) + \lambda \int_{0}^1 |f^{(r)}(t)|^2 \dd t \right],
\end{align}
where $\rho:\mathbbm{R} \to \mathbbm{R}_{+}$ is a loss function, e.g., $\rho(x) = x^2$ or $\rho(x) = |x|$, and $\lambda \geq 0$ is the penalty parameter. The roughness penalty is in terms of the integrated squared $r$th derivative of each candidate $f \in S_{K}^p([0,1])$ and is well-defined for $r<p$. This penalty is referred to as the O-spline penalty \citep[see, e.g.][]{Claeskens:2009}. Common choices for $p$ and $r$ are $p=4$ and $r=2$ leading to splines of order $4$ (degree $3$) with second derivative penalization. For $\lambda \to \infty$, the penalty term dominates forcing $\widehat{\mu}_n$ to lie in its null space. In such a situation, $\widehat{\mu}_n$ will be a polynomial of order $r$: smooth but typically far from the  observed data $Y_{ij}$. On the other hand, for $\lambda = 0$ the penalty vanishes, which typically leads to a $\widehat{\mu}_n$ that is close to the $Y_{ij}$ but noticeably less smooth.

Although the asymptotic theory for penalized spline estimators is well developed, the seemingly basic question of finite-sample existence has received little attention.  Proposition~\ref{prop:1} addresses this gap by giving a simple, easily checkable condition on the design points $T_{ij}$ and knots $x_j$ that ensures existence of the estimator for any continuous, coercive loss function. While stated here with functional data in mind, it applies equally well to a wide range of nonparametric estimation problems involving splines. We write  $\mathcal{T} = (T_{11}, \ldots, T_{1m_1}, \ldots, T_{nm_n})$ and  $N = \sum_{i=1}^{n} m_i$ in the statement of Proposition~\ref{prop:1} below.

\begin{prop}
\label{prop:1}
Suppose that $\rho(x)$ is a continuous, nondecreasing, and unbounded function of $|x|$, $p>r \geq 1$, $K +p \leq N$, $\lambda \geq 0$ and there exists a subset of $K+p$ distinct values from $\mathcal{T}$: $T_{s_1} < T_{s_2} < \ldots < T_{s_{K+p}}$ such that
\begin{align*}
B_{j,p}(T_{s_j}) > 0, \quad (j=1, \ldots, K+p).
\end{align*}
Then, the minimizer $\widehat{\mu}_n \in S_{K}^p([0,1])$ exists and can be chosen measurable.
\end{prop}
\begin{proof}
Write $f = \sum_{j=1}^{K+p} f_j B_{j,p}$. We show that 
\begin{align}
\label{eq:interp}
f(T_i) = 0, \quad \text{for all} \quad (i=1, \ldots, N) \implies f \equiv 0,
\end{align}
which is equivalent to the $N \times (K+p)$-dimensional matrix with elements  $B_{j,p}(T_i)$ having full column rank. 

Now, \eqref{eq:interp}  implies that $f$ interpolates $0$ at all $T_i \in \mathcal{T}$ and, by assumption, there exist $T_{s_1} < T_{s_2} < \ldots < T_{s_{K+p}}$ with the property that $B_{j,p}(T_{s_j}) > 0$. The Schoenberg-Whitney theorem \citep[Chapter XIII, Theorem (2)]{DB:2000} then guarantees the existence of a \textit{unique} $g \in S_{K}^p([0,1])$ satisfying $g(T_{s_1}) = \ldots = g(T_{s_{K+p}})=0$. Since $f$ vanishes at these $T_{s_j}$, it must be that $f \equiv g$, by the uniqueness of $g$. But $0 \in S_{K}^p([0,1])$ and the zero function trivially vanishes at the $T_{s_j}$. Hence $f \equiv g \equiv 0$ and it follows that the $N \times (K+p)$-dimensional matrix with elements $B_{j,p}(T_i)$ has full column rank.

To deduce the existence of a minimizer $\widehat{\mu}_n$ one may now argue as in the proof of Theorem 10.14 in \citep{Mar:2019} using the continuity and coercivity of $\rho$, namely $\rho(x) \to \infty$ as $|x| \to \infty$. The measurability of the minimizer $\widehat{\mu}_n$ follows from the fact that under the continuity of $\rho$ the objective function $L_n: \Omega \times S_{K}^p([0,1])$:
\begin{align*}
L_n(\omega,f) = \frac{1}{n} \sum_{i=1}^n \frac{1}{m_i} \sum_{j=1}^{m_i} \rho\left(Y_{ij}(\omega)-f(T_{ij}(\omega)) \right) + \lambda \int_{0}^1 |f^{(r)}(t)|^2 \dd t,
\end{align*}
is a Carath\'eodory function: $\omega \mapsto L_n(\omega,f)$ is measurable for every $f \in S_{K}^p([0,1])$ and $f \mapsto L_n(\omega,f)$ is continuous for every $\omega \in \Omega$, by the finite-dimensionality of $S_{K}^p([0,1])$. In the terminology of \citep[Example 14.29][]{Rock:2009}, $L_n$ is a \textit{normal integrand}. Hence, by Theorem 14.37 of these authors, the argmin $\widehat{\mu}_n$ is measurable.
\end{proof}

The condition $B_{j,p}(T_{s_j})>0$ for some $T_{s_j}$ requires that for each B-spline there exists a design point in the interior of its support, i.e., $x_j < T_{s_j} < x_{j+p}$. It is a truly minimal condition will be satisfied in all but degenerated situations in which the knots are tightly clustered and the $T_{ij}$ are concentrated within a narrow subregion of $[0,1]$. It should also be noted that Proposition~\ref{prop:1} covers all $\lambda \in [0,\infty)$ and so one does not need to argue differently for $\lambda = 0$ and $\lambda>0$. Its non-asymptotic nature, combined with its broad applicability to both convex and non-convex loss functions, makes it particularly useful for a wide range of practical scenarios.

\section{Asymptotic Theory}
\subsection{Preliminaries}
\label{sec:3}

Our next goal is to understand the asymptotic behaviour of $\widehat{\mu}_n$ and, in particular, how fast it recovers the true mean function $\mu_0$, as $n \to \infty$.  Contrary to the proofs of \citet{Xiao:2020} for the special case $\rho(x) = x^2$, our analysis is complicated by the lack of closed form solutions  under general $\rho$. Hence, we need to argue in terms of the objective function:
\begin{align}
\label{eq:objf}
L_n(f) = \frac{1}{n} \sum_{i=1}^n \frac{1}{m_i} \sum_{j=1}^{m_i} \rho\left(Y_{ij}-f(T_{ij}) \right) + \lambda \int_{0}^1 |f^{(r)}(t)|^2 \dd t, \quad f \in S_{K}^p([0,1]).
\end{align}
To establish a crucial link between $\widehat{\mu}_n$ and $L_n$ we require the following mild assumption on $\rho$.

\begin{enumerate} [label=(A\arabic*), ref=(A\arabic*)]
\item \label{A1} $\rho: \mathbbm{R} \to \mathbbm{R}_{+}$ is convex and absolutely continuous with derivative $\psi = \rho^{\prime}$ existing almost everywhere.
\end{enumerate}
Assumption \ref{A1} ensures that $L_n$, being the sum of two convex functions, is convex on $S_{K}^p([0,1])$. It covers many interesting loss functions, such as $\rho(x)=|x|$, $\rho(x) = x^2$ and the Huber loss $\rho_k(x) = x^21(|x| \leq k) + (2k|x| - k^2)1(|x|>k)$ for $k>0$, offering a compromise between the absolute and square losses.

It is convenient to endow $S_{K}^p([0,1])$ with the penalty-weighted inner product:
\begin{align*}
\langle f, g \rangle_{r,\lambda} = \langle f, g \rangle + \lambda\langle f^{(r)}, g^{(r)} \rangle,
\end{align*}
where $\langle \cdot, \cdot \rangle$ denotes the standard $\mathcal{L}^2([0,1])$ inner product. The associated norm is denoted with $\|\cdot\|_{r,\lambda}$. For $1 \leq r < p$ the inner product $\langle \cdot, \cdot \rangle_{r,\lambda}$ is well-defined on $S_{K}^p([0,1])$ and the associated norm $\|\cdot \|_{r,\lambda}$ is a valid norm: $\|\cdot\|_{r,\lambda}$ is homogeneous, satisfies the triangle inequality and, if $\|g\|_{r,\lambda} = 0$ for some $g \in S_{K}^p([0,1])$, then $g = 0$. Moreover, $S_{K}^p([0,1])$ is finite-dimensional so equipped with $\langle \cdot, \cdot\rangle_{r,\lambda}$ becomes a Hilbert space. But more can be said: all linear operators (including evaluation functionals) on finite-dimensional spaces are bounded, hence $S_{K}^p([0,1])$ is even a reproducing kernel Hilbert space (RKHS). This means that:
\begin{itemize}
\item[(i)] There exists a symmetric function (the reproducing kernel) $\mathcal{R}_{p, r,\lambda}:[0,1]^2 \to \mathbbm{R}$ such that $x \mapsto \mathcal{R}_{p,r,\lambda}(x,y) \in S_{K}^p([0,1])$ for every $y \in [0,1]$.
\item[(ii)] For every $f \in S_{K}^p([0,1])$, $f(x) = \langle f, \mathcal{R}_{p,r,\lambda}(x,\cdot) \rangle_{r,\lambda}$.
\end{itemize}

Lemma~\ref{Lem:RK} provides a characterization of $\mathcal{R}_{p,r,\lambda}$, which will be needed in the proof of our main asymptotic result (Theorem \ref{Thm:1}).

\begin{lemma}
\label{Lem:RK}
Suppose $r<p$ and $\lambda \geq 0$. The following assertions hold:
\begin{enumerate}[label=\arabic*.,ref=\arabic*]
\item \label{itm:rk1}  There exist $\{e_j\}_{j=1}^{K+p} \in S_{K}^p([0,1])$ and $\{\gamma_j\}_{j=1}^{K+p} \geq 0$ with the property that
\begin{align*}
\langle e_i,e_j\rangle = \delta_{ij}, \quad \langle e_i^{(r)},e_j^{(r)}\rangle = \gamma_j \delta_{ij}.
\end{align*}
The functions $\{e_j/(1+\lambda \gamma_j)^{1/2}\}_{j=1}^{K+p} \in S_{K}^p([0,1])$ form an orthonormal basis for $S_{K}^p([0,1])$.
\item \label{itm:rk2} The reproducing kernel $\mathcal{R}_{p,r,\lambda}(x,y)$ may be written as
\begin{align*}
\mathcal{R}_{p,r,\lambda}(x,y) = \sum_{j=1}^{K+p} \frac{e_j(x)e_j(y)}{1+\lambda \gamma_j}.
\end{align*}
\item \label{itm:rk3} For every $f \in \mathcal{L}^2([0,1])$:
\begin{align*}
\int_{0}^1 \int_{0}^1 \mathcal{R}_{p,r,\lambda}(x,y) f(x) f(y) \dd x \dd y \leq \int_{0}^1 |f(x)|^2 \dd x.
\end{align*}
\end{enumerate}
\end{lemma}

\begin{proof}

Part \ref{itm:rk1} of the lemma follows from standard simultaneous diagonalization results of a positive definite and a positive semidefinite operator \citep[see, e.g.][Theorem 4.8.1]{Hsing:2015}. As $S_{K}^p([0,1])$ is finite-dimensional, operators on $S_{K}^p([0,1])$ can be identified with matrices. In particular, consider the two $K+ p \times K+p$ matrices consisting of the inner products $\langle B_{i,p}, B_{j,p} \rangle$ and the inner products $\lambda \langle B_{i,p}^{(r)}, B_{j,p}^{(r)} \rangle$, respectively. The former matrix is positive definite, as the B-splines $B_{1,p} \ldots, B_{K+p,p}$ form a basis for $S_{K}^p([0,1])$, and the latter matrix is positive semidefinite. The final assertion of Part \ref{itm:rk1} follows from a linear algebra argument: if $V$ is a finite-dimensional vector space and $W \subset V$ has $\dim V$ linearly independent elements, then $\Span W=V$.

To verify Part \ref{itm:rk2}, recall that $x \mapsto \mathcal{R}_{p,r,\lambda}(x,y) \in S_{K}^p([0,1])$ for every $y \in [0,1]$. Hence, applying the result of Part \ref{itm:rk1}, there exist real-valued coefficients depending on $y$, $\{a_j(y)\}_{j=1}^{K+p}$, such that
\begin{align*}
\mathcal{R}_{p,r,\lambda}(x,y) = \sum_{j=1}^{K+p} a_j(y) \frac{e_j(x)}{(1+\lambda \gamma_j)^{1/2}}.
\end{align*}
Taking inner products of $\mathcal{R}_{p,r,\lambda}(x,y)$ with each of $\{e_j\}_{j=1}^{K+p}$, using the reproducing property (ii), the orthogonality of $\{e_j\}_{j=1}^{K+p}$ and $\langle e_j, e_j \rangle_{r,\lambda} = (1+\lambda \gamma_j)$ we obtain
\begin{align*}
e_j(y) = \langle \mathcal{R}_{p,r,\lambda}(y,\cdot), e_j \rangle_{r,\lambda} = \frac{a_j(y)}{(1+\lambda \gamma_j)^{1/2}} \langle e_j, e_j \rangle_{r,\lambda}  = a_j(y)(1+\lambda \gamma_j)^{1/2}.
\end{align*}
Hence $a_j(y) = e_j(y)/(1+\lambda \gamma_j)^{1/2}$ and Part \ref{itm:rk2} follows.

Finally, to establish Part \ref{itm:rk3} use the representation of $\mathcal{R}_{p,r,\lambda}$ obtained in Part \ref{itm:rk2} to see that
\begin{align*}
\int_{0}^1 \int_{0}^1 \mathcal{R}_{p,r,\lambda}(x,y) f(x) f(y) \dd x \dd y  & = \sum_{j=1}^{K+p} \frac{|\langle f, e_j \rangle|^2}{1+\lambda \gamma_j}
\\ & \leq \sum_{j=1}^{K+p} |\langle f, e_j \rangle|^2, && (\lambda \gamma_j \geq 0),
\\ & \leq \|f\|^2, && \text{(Bessel's inequality in $\mathcal{L}^2([0,1])$)}.
\end{align*}
\end{proof}

We need to use the objective function $L_n$ in order to argue that $\widehat{\mu}_n$ converges to $\mu_0$, but $L_n$ is defined on the spline subspace $S_{K}^p([0,1])$ and, in general, $\mu_0 \notin S_{K}^p([0,1])$. Nevertheless, if $\mu_0 \in \mathcal{C}^j([0,1])$ and $p>j$ there exists a spline approximation, $s_{\mu_0} \in S_{K}^p([0,1])$, which improves as $K \to \infty$ \citep[see, e.g.][Theorem XII (6)]{DB:2000}. We shall use this abstract spline approximation as a means of assessing proximity of $\widehat{\mu}_n$ to $\mu_0$.

\begin{lemma}
\label{Lem:1}
Assume that assumption \ref{A1} is satisfied. Then, for every sequence $\{C_n\}$ and  $D >0$ it holds that
\begin{align*}
\mathbb{P}\left( \|\widehat{\mu}_n-s_{\mu_0}\|_{r,\lambda}< D C_n^{1/2}  \right) \geq \mathbb{P}\left( \inf_{f \in S_{K}^p([0,1]): \|f\|_{r,\lambda} = D} L_n(s_{\mu_0}+C_n^{1/2}f) > L_n(s_{\mu_0}) \right).
\end{align*}
\end{lemma}
\begin{proof}

Denote the event
\begin{align*}
A := \{ \inf_{f \in S_{K}^p([0,1]): \|f\|_{r,\lambda} = D} L_n(s_{\mu_0}+C_n^{1/2}f) > L_n(s_{\mu_0})\},
\end{align*}
and define the open ball
\begin{align*}
O := \left\{ f \in S_K^p([0,1]) : \|f - s_{\mu_0}\|_{r,\lambda} < D C_n^{1/2} \right\}.
\end{align*}
It will be shown that on  $A$, every global minimizer of \( L_n \) lies in \( O \). Since $ \widehat{\mu}_n$ is a global minimizer, this will imply \( A \subset \{ \|\widehat{\mu}_n - s_{\mu_0} \|_{r,\lambda} < D C_n^{1/2} \} \), and the desired inequality would follow by the monotonicity of $\mathbb{P}$ relative to set containment.

Consider first the closed ball $\bar{O}=\{f \in S_{K}^p([0,1]):\|f-s_{\mu_0}\|_{r,\lambda} \leq DC_n^{1/2}\}$. Since $S_{K}^p([0,1])$ is finite-dimensional, $\bar{O}$ is compact. Moreover, the finite dimensionality of $S_{K}^p([0,1])$ and the convexity of $L_n$ implied by assumption \ref{A1}, jointly imply that $L_n$ is also continuous with respect to convergence in $\|\cdot\|_{r,\lambda}$. By the extreme value theorem, there exists a minimizer of $L_n$ in $\bar{O}$, $f_{\bar{O}}$. This $f_{\bar{O}}$ may be written as
\begin{align*}
f_{\bar{O}} = s_{\mu_0}+ r C_n^{1/2} f,
\end{align*}
for some $r \in [0, 1]$ and some $f \in S_{K}^p([0,1])$ satisfying $\|f\|_{r,\lambda}=D$. It must be that $r<1$ for, if that were not true, we would have 
\begin{align*}
L_n(f_{\bar{O}}) = L_n(s_{\mu_0}+ C_n^{1/2} f) > L_n(s_{\mu_0}),
\end{align*}
contradicting the definition of the minimizer: $L_n(f_{\bar{O}}) \leq L_n(f)$ for every $f \in \bar{O}$. It follows that $r<1$ and so $f_{\bar{O}} \in O \subset \bar{O}$.

We now claim that for every $g \in S_{K}^p([0,1]) \setminus \bar{O}$ we have $L_n(g)>L_n(s_{\mu_0})$. Indeed, consider the linear segment joining $g$ and $s_{\mu_0}$:
\begin{align*}
w(t) = (1-t)g + t s_{\mu_0}, \quad t \in [0,1].
\end{align*}
We have $w(0) = g$ and $w(1) = s_{\mu_0}$ and since $\|g - s_{\mu_0}\|_{r,\lambda} > D C_n^{1/2}$, the intermediate value theorem guarantees the existence of a $ t_0 \in (0,1)$ such that $\|w(t_0)-s_{\mu_0}\|_{r,\lambda}=DC_n^{1/2}$. Thus, $w(t_0) = s_{\mu_0} + C_n^{1/2} f_0$ for some $f_0 \in S_{K}^p([0,1])$ satisfying $\|f_0\|_{r,\lambda}=D$. On the event $A$ we have $L_n(s_{\mu_0})<L_n(s_{\mu_0} + C_n^{1/2} f_0)$, hence, by the convexity of $L_n$,
\begin{align*}
L_n(s_{\mu_0})<L_n(s_{\mu_0} + C_n^{1/2} f_0)=L_n(w(t_0)) \leq (1-t_0)L_n(g) + t_0 L_n(s_{\mu_0}), 
\end{align*}
or, equivalently, after collecting terms and dividing by $(1-t_0)>0$, $L_n(s_{\mu_0})<L_n(g)$. Thus, any \( g \notin \bar{O} \) cannot be a global minimizer.

Combining the above observations, we have shown that on the event $A$ every global minimizer must lie inside the open ball $O$. In particular, $\widehat{\mu}_n \in O$ and hence
\begin{align*}
\mathbb{P}\left( \|\widehat{\mu}_n-s_{\mu_0}\|_{r,\lambda}< D C_n^{1/2}  \right) \geq \mathbb{P}\left( \inf_{f \in S_{K}^p([0,1]): \|f\|_{r,\lambda} = D} L_n(s_{\mu_0}+C_n^{1/2}f) > L_n(s_{\mu_0}) \right),
\end{align*}
as claimed.
\end{proof}

Like Proposition~\ref{prop:1}, Lemma~\ref{Lem:1} is non-asymptotic and holds for every $n \in \mathbbm{N}$. To illustrate its usefulness, take a sequence $C_n \to 0$. If for every $\epsilon>0$ there exists a $D = D(\epsilon)$ such that
\begin{align}
\label{eq:plim}
\liminf_{n \to \infty} \mathbb{P}\left( \inf_{f \in S_{K}^p([0,1]): \|f\|_{r,\lambda} = D} L_n(s_{\mu_0}+C_n^{1/2}f) > L_n(s_{\mu_0}) \right) \geq 1-\epsilon,
\end{align}
Lemma~\ref{Lem:1} guarantees $\|\widehat{\mu}_n-s_{\mu_0}\| \leq \|\widehat{\mu}_n-s_{\mu_0}\|_{r,\lambda} =O_{\mathbb{P}}(C_n^{1/2})$. The triangle inequality then gives
\begin{align*}
\|\widehat{\mu}_n-\mu_0\| \leq O_{\mathbb{P}}(C_n^{1/2}) + \|s_{\mu_0}-\mu_0\|.
\end{align*}
The second term on the right is the deterministic spline approximation error, which, if $\mu_0 \in \mathcal{C}^j([0,1])$ and under typical assumptions, decays at the rate $O(K^{-j})$. This leads to $\|\widehat{\mu}_n-\mu_0\| = O_{\mathbb{P}}(C_n^{1/2}+K^{-j})$ and so, if $K \to \infty$, $\widehat{\mu}_n$ is a consistent estimator with rate of convergence $O_{\mathbb{P}}(C_n^{1/2}+K^{-j})$.

\subsection{Main Result}

To obtain the best possible rate of convergence we introduce and discuss the following set of assumptions. To lighten the notation, we adopt the abbreviation
\begin{align*}
\epsilon_{i}(t) := \left(X_i(t) - \mu_0(t)\right)+\zeta_i(t), \quad (i=1, \ldots, n),
\end{align*}
and write $\epsilon_{ij} = \epsilon_i(T_{ij})$. We call $\epsilon_i(\omega,\cdot ) = (X_i(\omega,\cdot) - \mu_0(\cdot))+\zeta_i(\omega,\cdot)$ the \textit{error process}. 

\begin{enumerate} [label=(A\arabic*), ref=(A\arabic*)]
\setcounter{enumi}{1}
\item \label{A2} Let $h_i = x_i-x_{i-1}$ with $x_i$ the knots and $h =\max_i h_i$. Then, $\max_i |h_{i+1}-h_i|= o(K^{-1})$ and there exists an $M>0$ such that $h/\min_i h_i \leq M$. 
\item \label{A3} The error processes $\epsilon_i(\omega, t) : (\Omega, \mathcal{A}, \mathbb{P}) \times [0,1] \to \mathbbm{R}$, $i =1, \ldots, n,$ are independent and identically distributed.
\item \label{A4}  The discretization variables $T_{ij}$ are independent and identically distributed with density $q$ satisfying  $0<c \leq q(t) \leq C < \infty$ for some $0<c\leq C<\infty$ for all $t \in [0,1]$. Moreover, the $T_{ij}$ are independent of the error processes $\epsilon_{i}$. 
\item \label{A5} There exist finite constants $\kappa$ and $M_1$ such that for all $|y| < \kappa$,
\begin{align*}
\left|\psi(x+y)-\psi(x) \right| \leq M_1, \quad \forall x \in \mathbbm{R}.
\end{align*}
\item \label{A6} For $u \to 0$, it holds that
\begin{align*}
\sup_{t \in [0,1]}\mathbb{E} [\left|\psi(\epsilon_1(t)+u)-\psi(\epsilon_1(t))\right|^2 ] = O(|u|).
\end{align*}
\item \label{A7} $\mathbb{E}[\psi(\epsilon_{1}(t))] =0$ for all $t \in [0,1]$, \ $\sup_{t \in [0,1]} \mathbb{E}[|\psi(\epsilon_{1}(t)|^2] < \infty$ and there exists a function $\delta: [0,1] \to \mathbbm{R}_{+}$ satisfying  $0< \inf_{t \in [0,1]} \delta(t) \leq \sup_{t \in [0,1]}  \delta(t) <\infty$ and
\begin{align*}
\sup_{t \in [0,1]}\left| \mathbb{E}[\psi(\epsilon_{1}(t)+u)]-\delta(t) u \right| = o(u),
\end{align*}
as $u \to 0$.
\end{enumerate}

Assumption \ref{A2} is a knot quasi-uniformity condition, which has been extensively used in the asymptotics of lower rank spline estimators at least since \citep{Shen:1998}. Assumption \ref{A3} ensures that we have $n$ independent units and is a considerable generalization of previous assumptions in the literature requiring also independence between the $X_i$ and the $\zeta_i$, see, e.g., the statement of Theorem 3.2 in \citep{Cai:2011} or Assumption 6 in \citep{Xiao:2020}. Assumption \ref{A4} reflects the longitudinal nature of the data: each subject is observed repeatedly at random times drawn independently from a common density $q$ on $[0,1]$ that is bounded away from zero and infinity; this means that all subregions are adequately sampled while its independence from the $\epsilon_i$ ensures that there is no informative sampling. Assumption \ref{A5} is very mild and is satisfied by the derivatives or subgradients of all commonly used convex loss functions.

Assumptions \ref{A6} and \ref{A7} are more technical, but ensure that our results cover both smooth and non-smooth loss functions. Specifically,
assumption \ref{A6} is a mean-square Lipschitz continuity condition while assumption \ref{A7} is a Fisher-consistency condition ensuring that we are indeed estimating the mean function $\mu_0$ instead of another location function. If $\psi$ is Lipschitz, e.g., $\psi(x) = x$ or $\psi_k(x) = \max(-k,\min(x,k))$ for some $k>0$, then \ref{A6} is satisfied with the stronger $O(|u|^2)$ on the right side. If $\psi$ is smooth, $0<\mathbb{E}[\psi^{\prime}(\epsilon_1(t))] <\infty$ uniformly in $t \in [0,1]$ and $\|\psi^{\prime \prime}\|_{\infty}<\infty$, then  we can take $\delta(t) = \mathbb{E} [\psi^{\prime}(\epsilon_1(t))]$ ensuring the validity of the second part of \ref{A7}. For $\psi(x) = x$, we can simply take $\delta(t) := 1$. 

For less smooth loss functions assumptions \ref{A6} and \ref{A7} typically require additional conditions on the distribution of the error process $\epsilon_1(t)$. These assumptions can still be verified on a case-to-case basis, as we now demonstrate by means of a prominent example. Consider the loss function $\rho(x) = |x|$ with subgradient $\psi(x) = \sign(x)$ for $x \neq 0$ and $ \psi(0) = 0$.  If at each $t \in [0,1]$, $\epsilon_1(t)$ has CDF $F_{t}$ with positive and bounded Lebesgue density $f_{t}$ in a neighbourhood about zero, for $u \to 0$ we find
\begin{align*}
\sup_{t \in [0,1]}\mathbb{E} [\left|\psi(\epsilon_1(t)+u)-\psi(\epsilon_1(t))\right|^2 ] \leq 4 |u| (\sup_{t \in [0,1]} \sup_{|u|\leq \delta} f_{t}(u)),
\end{align*}
for every small $\delta>0$. If there exists a $\delta>0$ for which $\sup_{t \in [0,1]} \sup_{|u|\leq \delta} f_{t}(u) < \infty$, the right side is $O(|u|)$ and \ref{A6} holds. Furthermore, if $F_t(0) = 1/2$, that is, $\epsilon_1(t)$ has median at zero, then choosing $\delta(t) = 2 f_{t}(0)$, condition \ref{A7} is fulfilled provided that
\begin{align*}
\lim_{u \to 0} \sup_{t \in [0,1]} \left|\frac{F_{t}(u)-F_{t}(0)}{u} - f_{t}(0)  \right| = 0.
\end{align*}
That is, the condition is fulfilled if the family of CDFs $\{F_{t}, t \in [0,1]\}$ is uniformly first order differentiable at $0$ and $f_t(0)$, the derivatives (or densities) at $0$, are uniformly bounded away from $0$ and infinity. Thus, heavy tails are permitted. Conditions \ref{A6} and \ref{A7} can be verified in a likewise manner for many other loss functions.

Theorem~\ref{Thm:1} below contains our main result. In its statement we use $m$ to denote the harmonic mean of the $m_i$, that is, $m = n/(\sum_{i=1}^n m_i^{-1})$; this may be viewed a measure of density of the sampling points $T_{ij}$ within $[0,1]$. We sketch the proof and defer the technical lemmas to the Appendix at the end of the paper.

\begin{theorem}
\label{Thm:1}
Assume that assumptions \ref{A1}--\ref{A7} hold and that $\lambda \to 0$ and $K \to \infty$ in such a way that $n^{\delta-1}K^3 \to 0$ for some $\delta>0$, $K \min\{\lambda, \lambda^2 K^{2r}\} \to 0$ and $K (nm)^{-1} \min(\lambda^{-1/(2r)}, K) \to 0$ as $n \to \infty$. Then, if $\mu \in \mathcal{C}^{j}([0,1])$ with $r \leq j < p$ and $C_n := n^{-1} + (nm)^{-1} \min(\lambda^{-1/(2r)}, K) + K^{-2j} + \min(\lambda, \lambda^2 K^{2r})$, for every $\epsilon>0$ there exists a $D = D(\epsilon)$ such that
\begin{align*}
\liminf_{n \to \infty} \mathbb{P}\left( \inf_{f \in S_{K}^p([0,1]): \|f\|_{r,\lambda} = D} L_n(s_{\mu_0}+C_n^{1/2}f) > L_n(s_{\mu_0}) \right) \geq 1-\epsilon.
\end{align*}
\end{theorem}

\begin{proof}

Write for convenience $R_{ij} = \mu_0(T_{ij}) - s_{\mu_0}(T_{ij})$. By \ref{A1}, $\rho$ is absolutely continuous, hence we may decompose
\begin{align*}
L_n(s_{\mu_0}+C_n^{1/2}f) - L_n(s_{\mu_0}) =   I_1(f) + I_2(f) + I_3(f) + I_4(f),
\end{align*}
with
\begin{align*}
I_1(f) & = \frac{1}{n} \sum_{i=1}^n \frac{1}{m_i} \sum_{j=1}^{m_i} \mathbb{E}\left[ \int_{R_{ij}}^{R_{ij}+C_n^{1/2}f(T_{ij})} \{\psi\left(\epsilon_{ij}+u\right) - \psi(\epsilon_{ij})   \}\dd u \right] + \lambda C_n \left\|f^{(r)} \right\|^2,
\\  I_2(f) & = \frac{1}{n} \sum_{i=1}^n \frac{1}{m_i} \sum_{j=1}^{m_i} \int_{R_{ij}}^{R_{ij}+C_n^{1/2}f(T_{ij})}  \{\psi\left(\epsilon_{ij}+u\right) - \psi(\epsilon_{ij}) \}\dd u
\\ & \quad - \frac{1}{n} \sum_{i=1}^n \frac{1}{m_i} \sum_{j=1}^{m_i} \mathbb{E} \left[ \int_{R_{ij}}^{R_{ij}+C_n^{1/2}f(T_{ij})}  \{\psi\left(\epsilon_{ij}+u\right) - \psi(\epsilon_{ij}) \}\dd u\right] \\
I_3(f) & =  \frac{C_n^{1/2}}{n} \sum_{i=1}^n \frac{1}{m_i} \sum_{j=1}^{m_i} \psi(\epsilon_{ij}) f(T_{ij}) \\
 I_4(f) & =  2 \lambda C_n^{1/2} \langle f^{(r)}, s_{\mu_0}^{(r)} \rangle.
\end{align*}

To complete the proof we will show the following:
\begin{align}
\inf_{f \in S_{K}^p([0,1]):\|f\|_{r,\lambda} = D} I_1(f) & \geq c_0 D^2 C_n+ O(1) D C_n \label{asser:1} \\
\sup_{f \in S_{K}^p([0,1]):\|f\|_{r,\lambda} \leq D} |I_2(f)| &= o_{\mathbb{P}}(1) C_n \label{asser:2} \\
\sup_{f \in S_{K}^p([0,1]):\|f\|_{r,\lambda} \leq D} |I_3(f)| &= O_{\mathbb{P}}(1) D C_n \label{asser:3} \\ 
\sup_{f \in S_{K}^p([0,1]):\|f\|_{r,\lambda} \leq D} |I_4(f)| &= O(1) D C_n \label{asser:4},
\end{align}
for a strictly positive constant $c_0$. From \eqref{asser:1}--\eqref{asser:4} it follows that 
\begin{align*}
\inf_{f \in S_{K}^p([0,1]): \|f\|_{r, \lambda} = D} L_n(s_{\mu_0}+C_n^{1/2}f) - L_n(s_{\mu_0}) \geq c_0 D^2 C_n + O_{\mathbb{P}}(1) D C_n + o_{\mathbb{P}}(1) C_n.
\end{align*}
For every $\epsilon>0$ we can now find a sufficiently large $D > 1$ (so that $D^2 >D$) ensuring that the right side will be strictly positive with probability at least $1-\epsilon$ for all large $n$. Assertion \eqref{asser:4} may be established exactly as in the proof of Theorem~2 in \citep{Kal:2023} and its proof is omitted. The other assertions are proven in Lemmas \ref{Lem:I1}, \ref{Lem:I2} and \ref{Lem:I4}  in the Appendix.
\end{proof}

The conditions of Theorem \ref{Thm:1} require that $K$ increases with the sample size, but at a controlled rate. The condition $n^{\delta-1} K^3 \to 0$ can be replaced with $n^{\delta-1}K^2$ for some $\delta>0$ if $\psi$ is Lipschitz. Thus, for smoother loss functions, the number of knots can grow faster albeit still at a controlled rate. Such restrictions are in line with the main motivation for penalized spline estimators: a sufficiently but not excessively rich spline basis can capture most underlying patterns.  

 The optimal sequence $C_n$ may be derived by examining (i) the variance of the linearized estimator, which is captured by $I_3$, and (ii) the penalization bias, which is captured by $I_4$. Contrary to smoothing spline type full rank estimators, the variance and penalization bias herein are more intricate depending jointly on $K$ and $\lambda$. The other terms of our decomposition, $I_1$ and $I_2$, capture on the one hand the identifiability of the estimator and on the other the non-linear remainder of our expansion, which we show to be asymptotically negligible with an empirical process argument. For $\psi(x) = x$ the non-linear remainder vanishes identically and our proof greatly simplifies. In this case, we recover the result in \citep[Theorem 3.2]{Xiao:2020}, although our proof, using elements of functional instead of matrix analysis, remains notably different.

For this specific $C_n$,  $K^{-2j} \leq C_n$, so Theorem~\ref{Thm:1}, via Lemma~\ref{Lem:1}, implies
\begin{align}
\label{eq:rate}
\int_{0}^1 |\widehat{\mu}_n(t) -\mu_0(t)|^2 \dd t = O_{\mathbb{P}}\left(  n^{-1} + (nm)^{-1} \min(\lambda^{-1/(2r)}, K) + K^{-2j} + \min(\lambda, \lambda^2 K^{2r})\right).
\end{align}
This rather interesting rate of convergence sheds light on the interplay between the number of knots $K$, the penalty parameter $\lambda$ and the discretization of the curves, as captured by the harmonic mean $m$. To start with, for $K \geq \lambda^{-1/(2r)}$ the rate of convergence becomes
\begin{align*}
\int_{0}^1 |\widehat{\mu}_n(t) -\mu_0(t)|^2 \dd t = O_{\mathbb{P}}\left(  n^{-1} + (nm \lambda^{1/(2r)})^{-1} + K^{-2j} + \lambda \right),
\end{align*}
which is referred to as "\textit{a large knot scenario}". It is interesting to note that up to the spline approximation error $K^{-2j}$ is the rate of convergence of the smoothing spline estimators of \citep{Cai:2011} and \citep{Kal:2023b}. But while these estimators require knots at each unique $T_{ij}$, the penalized spline estimator only requires $K$ knots. To appreciate the difference in  magnitude, note that if all the $T_{ij}$ are distinct (which will be the case under assumption \ref{A4}) smoothing spline estimators are $\sum_{i=1}^{n}m_i$-dimensional and since $K<< \sum_{i=1}^{n}m_i$, penalized spline estimators will be enormously more computationally efficient. Moreover, as we demonstrate in Section \ref{sec:3}, these computational savings come at a very little (if any) accuracy cost.

On the other hand, if $K < \lambda^{-1/(2r)}$ then the right side of \eqref{eq:rate} becomes
\begin{align*}
\int_{0}^1 |\widehat{\mu}_n(t) -\mu_0(t)|^2 \dd t = O_{\mathbb{P}}\left(  n^{-1} + (nm)^{-1}K  + K^{-2j} +  \lambda^2 K^{2r}\right), 
\end{align*}
which is referred to as the "\textit{small knot scenario}". The interested reader is referred to \citep{Xiao:2020} for a detailed discussion of these situations. Here, we only note the remarkable fact that if $K$ and $\lambda$ grow and decay at appropriate rates, then the parametric rate of convergence $n^{-1}$ is attainable even with discretely observed functional data. To see this, in the large knot scenario assume that $m \geq n^{1/(2r)}$ and take $\lambda \approx (mn)^{-2r/(2r+1)}$ and $K \geq n^{1/(2j)}$. In the small knot scenario, assume again that $m \geq n^{1/(2r)}$ and take $\lambda = O(K^{-2r})$ and $n^{1/(2r)} \leq K \leq m$. These choices lead to $\|\widehat{\mu}_n-\mu_0\|^2 = O_{\mathbb{P}}(n^{-1})$ in both situations.

The critical breakpoint $K \approx \lambda^{-1/(2r)}$ in the asymptotic behaviour of penalized spline estimators was identified by \citet{Claeskens:2009} in nonparametric regression. It is interesting to see that that although Theorem~\ref{Thm:1} establishes a very different asymptotic behaviour than in nonparametric regression, it is the very same breakpoint that separates the large and small knot asymptotic regimes.  

\begin{figure}[H]
\centering
\subfloat{\includegraphics[width = 0.48\textwidth]{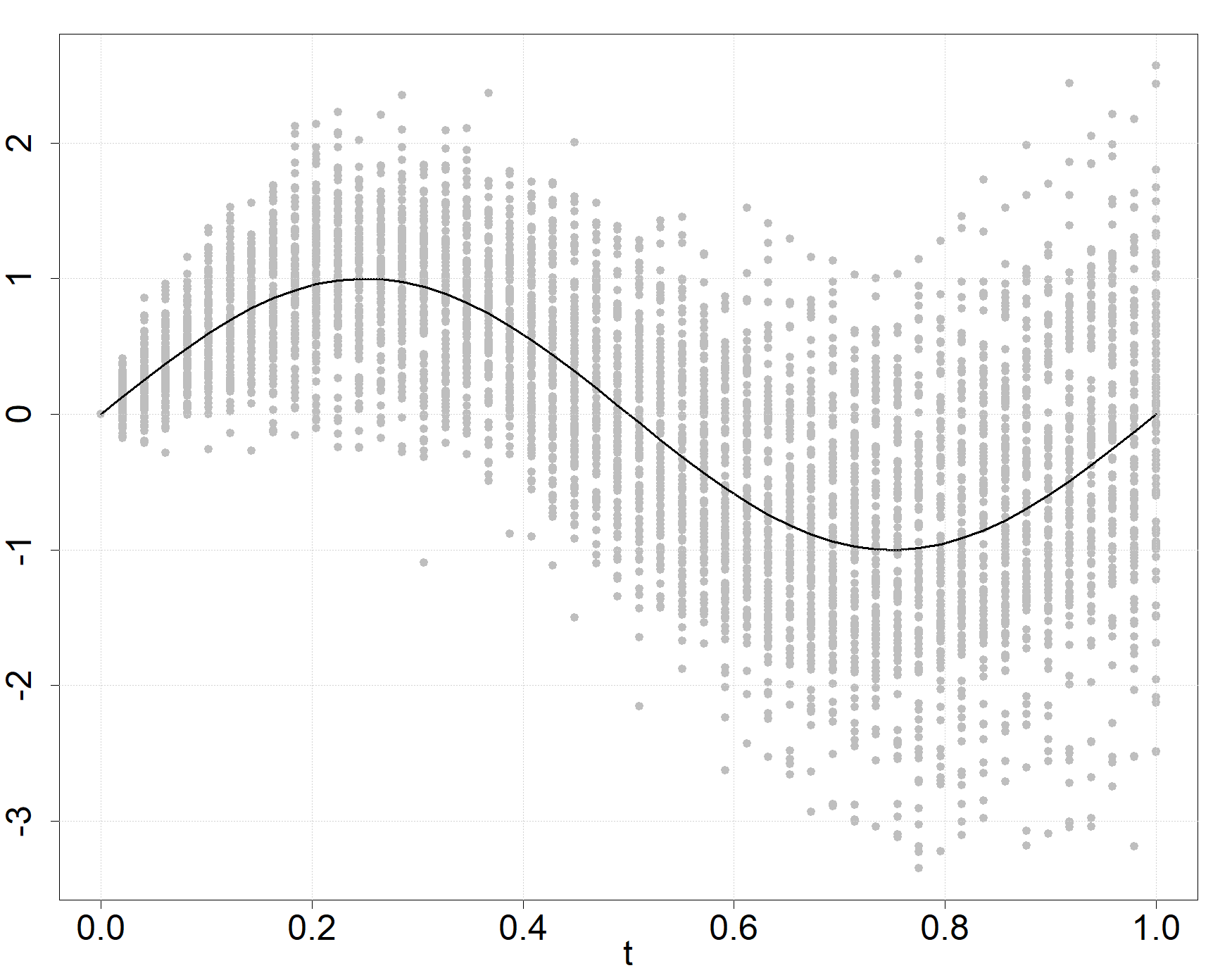}}
\subfloat{\includegraphics[width = 0.48\textwidth]{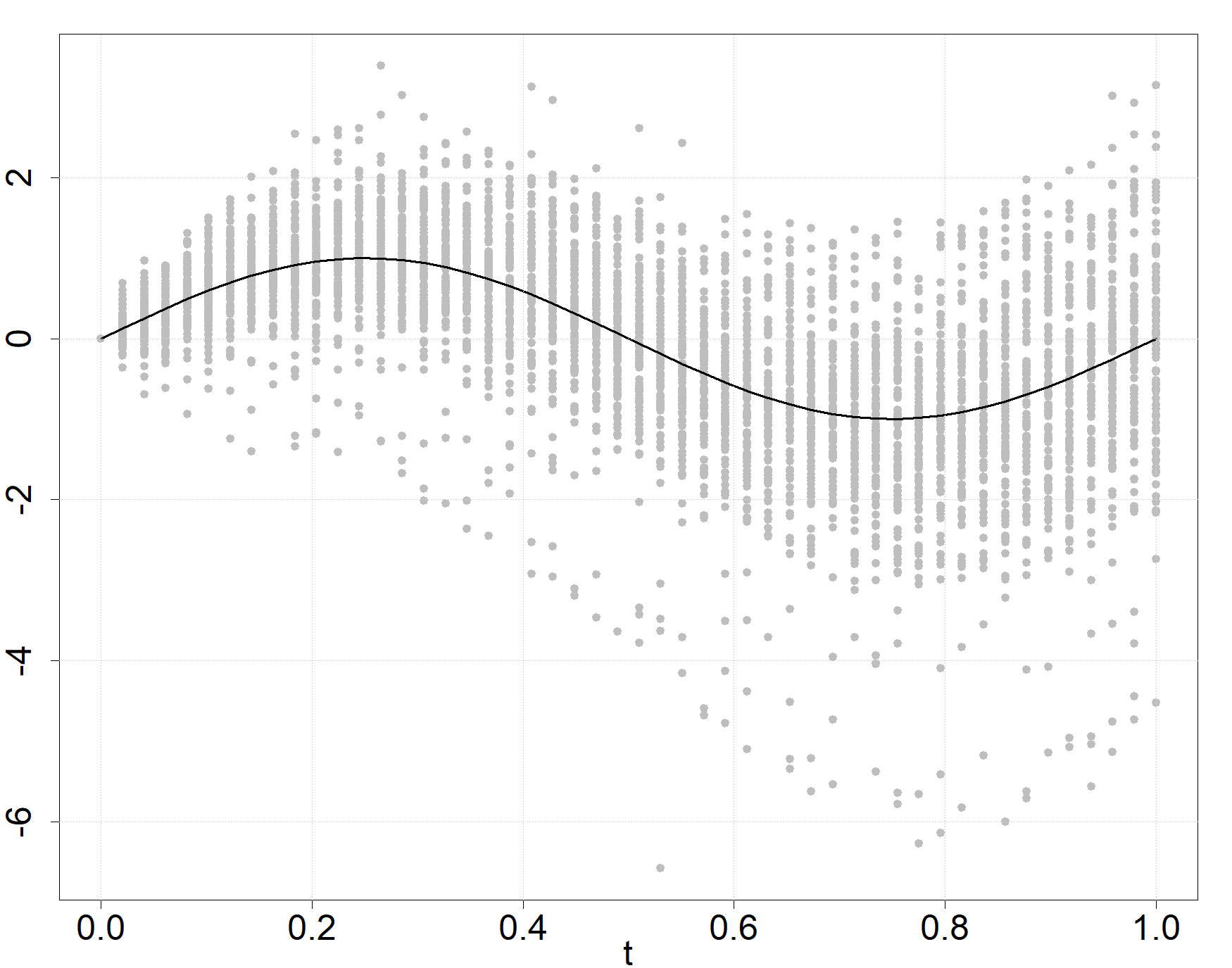}}
\caption{A sample of 100 curves $X_i$ discretized at $50$ points within $[0,1]$ generated with $Z_{ik} \sim N(0,1)$ (left) and with $Z_{ik} \sim t_5$ (right). The line (\full) depicts the true mean function $\mu_1(t) = \sin(2\pi t)$.}
\label{fig:curves}
\end{figure}

\section{Estimation Performance and Computing Times}
\label{sec:4}

In our numerical experiments, we are interested in assessing (i) estimation performance with respect to a variety of $X_i$ or $\zeta_i$ in \eqref{eq:discrfd} and (ii) the computational effort relative to the robust smoothing spline estimator of \citep{Kal:2023b}. In particular, the estimators to be compared are:
\begin{itemize}
\item The Penalized Spline Least-Squares Estimator of \citep{Xiao:2020} with
$\rho(x) = x^2$ in \eqref{eq:est},  abbreviated as PenLS.
\item The Penalized Spline Least Absolute Deviations Estimator with $
\rho(x) = |x|$ in \eqref{eq:est}, abbreviated as PenLAD.
\item The Smoothing Spline Least Absolute Deviations Estimator with $\rho(x) = |x|$ of \citep{Kal:2023b}, abbreviated as SmLAD.
\end{itemize}
We have fitted the penalized spline estimators PenLS and PenLAD with $p=4$, $r=2$ and $K=30$ equidistant knots within $[0,1]$. For the smoothing spline estimator SmLAD we have taken $r=2$, which also leads to spline estimators of order $4$ but with knots at all distinct $T_{ij}$. The robust estimators are computed with iteratively reweighted least-squares and for all estimators we have selected the penalty parameter $\lambda$ with weighted Generalized Cross Validation in the manner described in \citep{Kal:2023b}. To ensure that these estimators are directly comparable we have implemented them ourselves in C++ with an R-interface \citep{R:2025}; all related source functions along with the R-scripts reproducing our plots and tables may be found at \url{https://github.com/ioanniskalogridis/Robust-Functional-Location-Estimation}. 

To emulate the setting of discretely sampled functional data, we have generated curves $X_i$ on $[0,1]$ according to the truncated Karhunen-Lo\'eve decomposition:
\begin{align*}
X_i(t) = \mu_0(t) +  2^{1/2}\sum_{k=1}^{50} Z_{ik} \frac{\sin((k-1/2)\pi t)}{(k-1/2)\pi}, \quad (i=1, \ldots, 100),
\end{align*}
where $\mu_0(t)$ is either $\mu_1(t) = \sin(2 \pi t)$ or $\mu_2(t) = \exp[-(t-0.25)^2/0.01]+\exp[-(t-0.5)^2/0.01] + \exp[-(t-0.75)^2/0.01]$ representing a globally smooth function and a spiky function, respectively. The random variables $Z_{ij}$ are i.i.d. following the $t$-distribution with $5$ degrees of freedom (df); a $t$-distribution with $5$ df instead of the typically used Gaussian distribution leads to some curves that exhibit a different behaviour than the rest; see the left and right panels of Figure~\ref{fig:curves} for a visual comparison.

\begin{table}[H]
\centering
\begin{tabular}{llcccccc}
& & \multicolumn{2}{c}{PenLS} & \multicolumn{2}{c}{PenLAD} & \multicolumn{2}{c}{SmLAD} \\
\multicolumn{1}{c}{$\mu_0$} & df & Mean & SE & Mean & SE & Mean & SE \\ 
\toprule
\multirow{4}{*}{$\mu_1$} & 1 & 87006 & 44968 & 15.46 & 1.43 & 15.53 & 1.43 \\
                         & 2 & 9.04 & 0.40 & 6.29 & 0.23 & 6.33 & 0.23 \\
                         & 5 & 5.08 & 0.23 & 5.10 & 0.18 & 5.12 & 0.18 \\
                         & 10 & 4.79 & 0.20 & 5.00 & 0.17 & 5.03 & 0.17 \\
\midrule
\multirow{4}{*}{$\mu_2$} & 1 & 1824206 & 1216824 & 39.50 & 0.95 & 41.35 & 0.90 \\
                         & 2 & 12.86 & 0.36 & 8.65 & 0.20 & 9.62 & 0.21 \\
                         & 5 & 6.50 & 0.24 & 7.69 & 0.20 & 8.53 & 0.20 \\
                         & 10 & 5.85 & 0.20 & 6.79 & 0.19 & 7.50 & 0.19 \\
\bottomrule
\end{tabular}
\caption{Mean and standard errors $(\times 1000)$ of the mean-squared errors of the competing estimators based on $500$ replications with $n=100$ and $m_i \in [25,40]$ randomly.}
\label{tab:MSE}
\end{table}

We have evaluated each $X_i$ at $50$ equidistant points $T_j$ within $[0,1]$ and have subsequently randomly generated an integer $m_i$ in $[25,40]$ and randomly selected $m_i$ out of these $T_j$. We have only kept the values of the $X_i$ at these $m_i$ points thereby obtaining an array $\{X_{ij}\}_{i=1,j=1}^{n,m_i}$. Finally, we have generated the $Y_{ij}$ according to $Y_{ij} = X_{ij}+0.5\zeta_{ij},\ j=1, \ldots, m_i, i=1, \ldots, 100$, where the noise variables $\zeta_{ij}$ are i.i.d. following a t-distribution with df in $\{2,3,5,10\}$. After computing all estimates from the $Y_{ij}$ we evaluate them at the $T_j$ and assess performance with the mean-squared error:
\begin{align*}
\MSE = \frac{1}{50} \sum_{j=1}^{50} |\widehat{\mu}_n(T_j)-\mu_0(T_j)|^2,
\end{align*}
which is an approximation to the $\mathcal{L}^2([0,1])$ error. Table 1 presents the average MSEs (Mean) along with their standard errors (SE).
 
\begin{figure}[H]
\centering
\subfloat{\includegraphics[width = 0.48\textwidth]{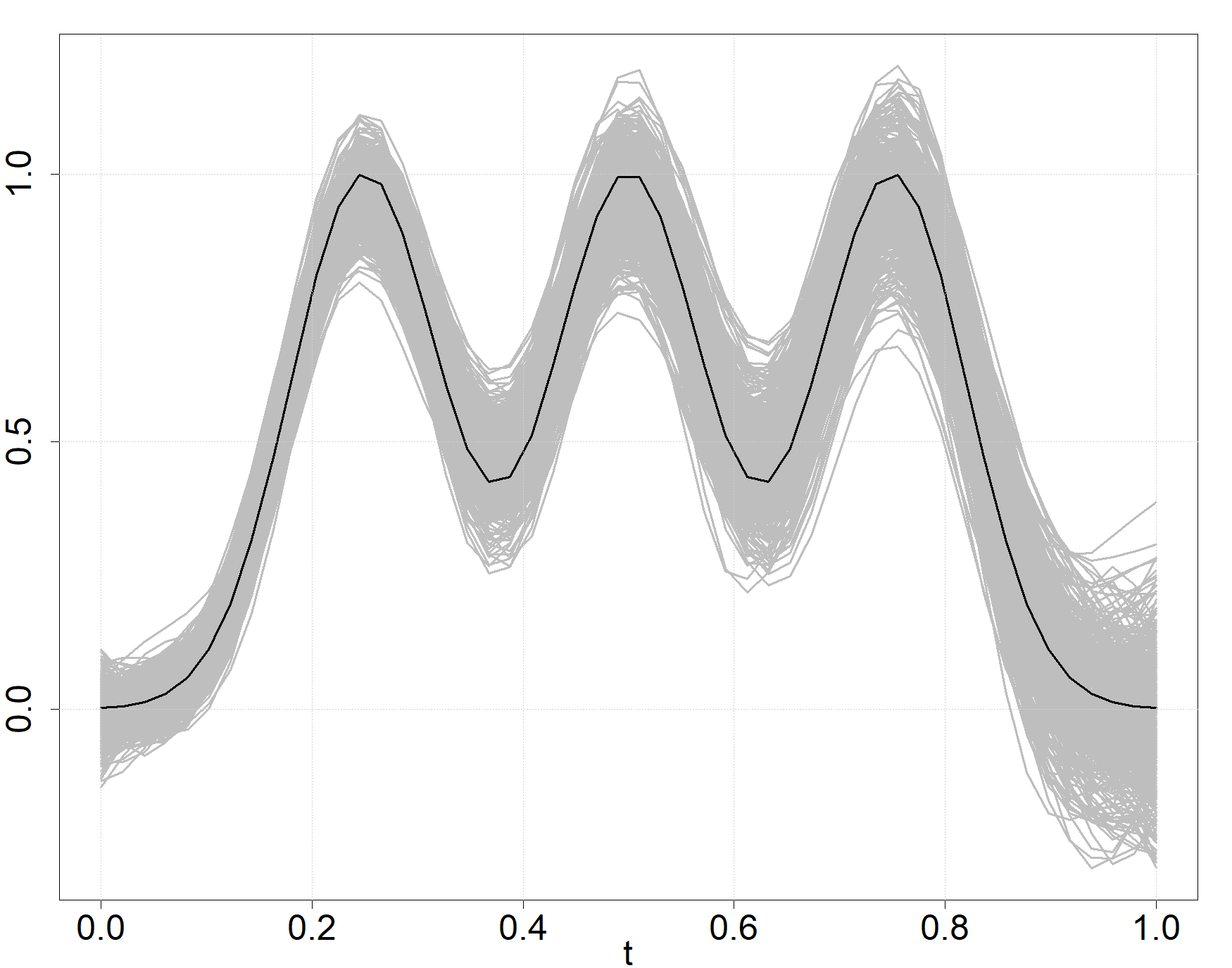}} \
\subfloat{\includegraphics[width = 0.48\textwidth]{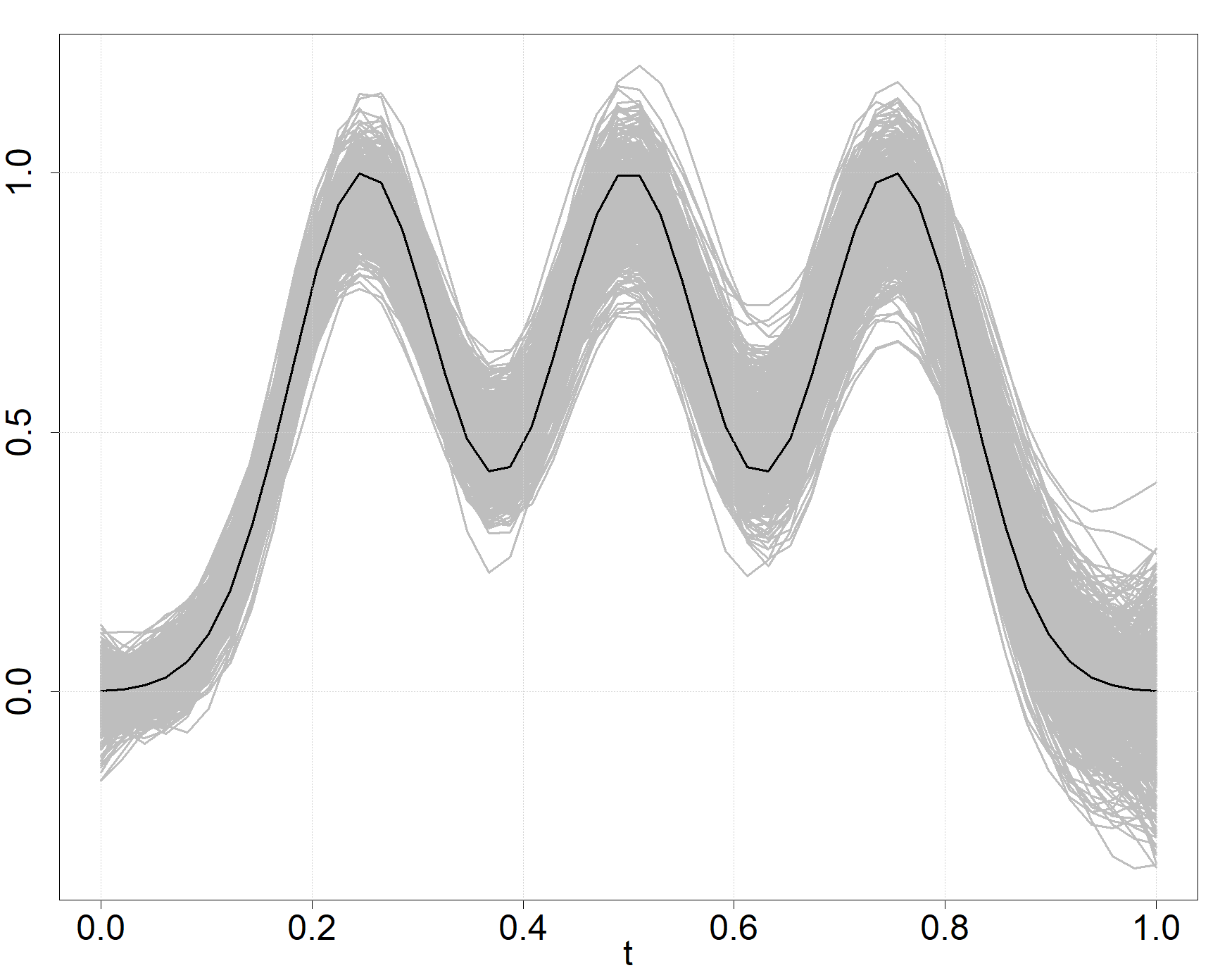}} \\
\subfloat{\includegraphics[width = 0.48\textwidth]{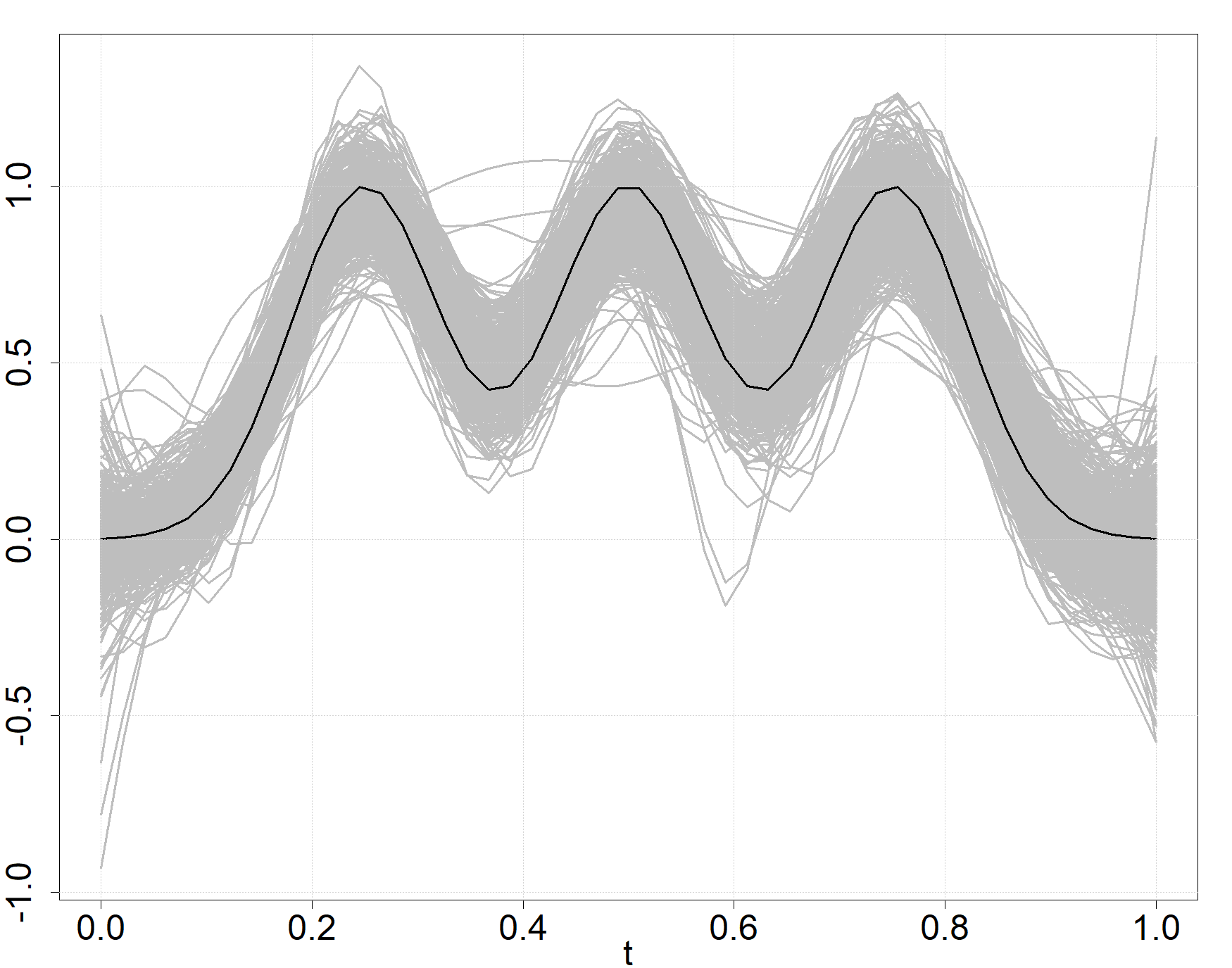}} \
\subfloat{\includegraphics[width = 0.48\textwidth]{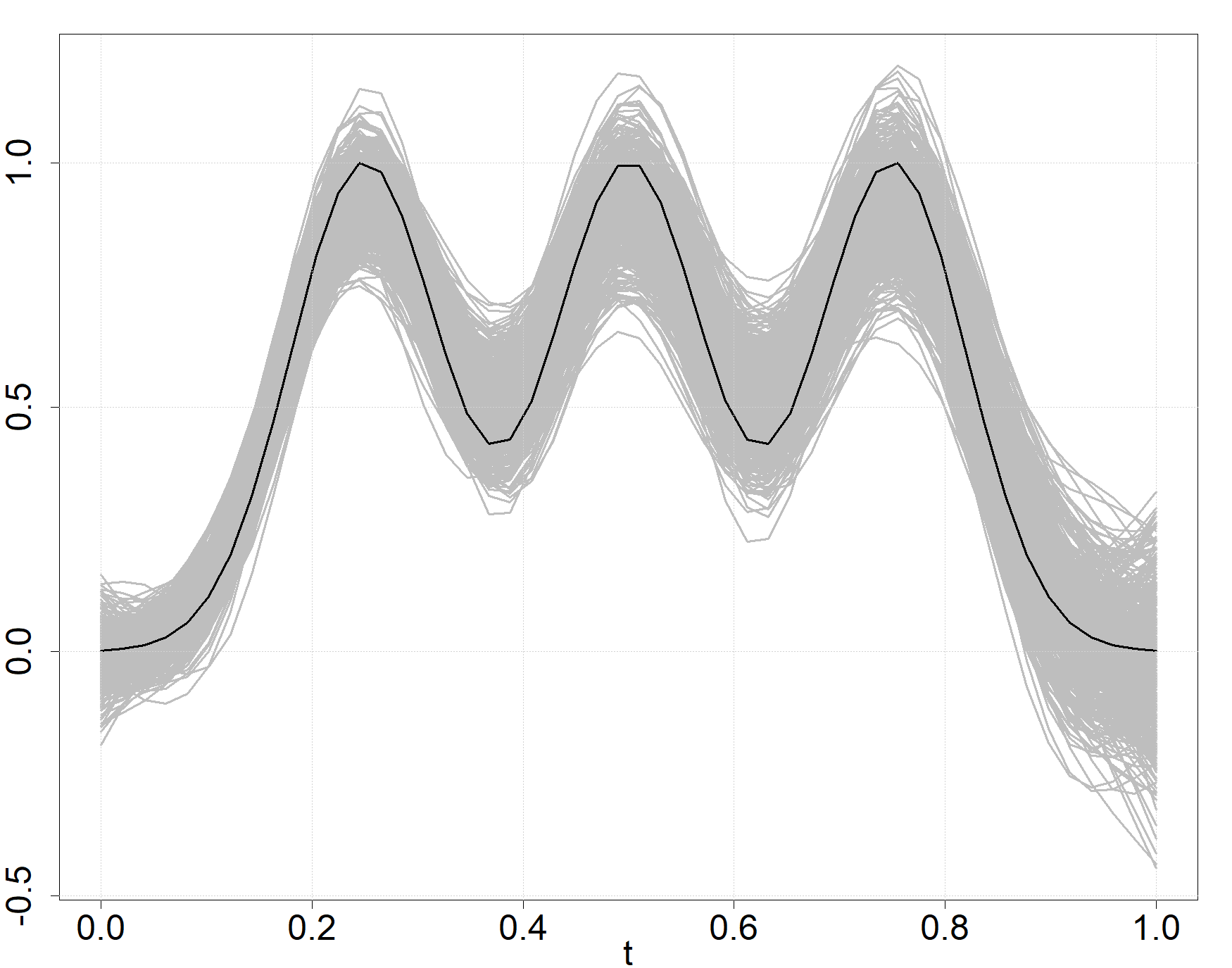}} 
\caption{Top row: 500 estimates for $\mu_2$  with $\zeta_i \sim t_{10}$  by PenLS (left) and PenLAD (right). Bottom row: 500 estimates for $\mu_2$ with $\zeta_i \sim t_2$ by PenLS (left) and PenLAD (right). The line (\full) depicts the true mean function $\mu_2$. }
\label{fig:comp}
\end{figure}

Table \ref{tab:MSE} leads to several interesting findings. Firstly, the least-squares PenLS tends to outperform the robust estimators PenLAD and SmLAD under relatively light-tailed $\zeta_i$ with df $\geq 5$, but the difference is small and barely visible in practice, cf. the top row of Figure \ref{fig:comp}. For smaller df, that is, for heavier tailed $\zeta_i$, the robust estimators perform noticeably better, see also the bottom row of Figure \ref{fig:comp}. It is interesting to observe that although their performance also decreases with heavier-tailed $\zeta_i$, the deterioration is very modest in comparison to PenLS, which becomes completely unreliable in such settings.

The performance of the robust estimators PenLAD and SmLAD is similar with respect to $\mu_1$, but PenLAD outperforms SmLAD with respect to the more challenging $\mu_2$. At first glance this may seem counterintuitive, since SmLAD uses more basis functions and therefore could be thought better at capturing local features. However,  Theorem~\ref{Thm:1} implies the existence of a saturation point beyond which additional basis functions no longer improve accuracy. Beyond this point, the use of unnecessarily large basis expansions can even worsen numerical stability, as larger matrices must be inverted. This effect is also visible in the boxplots of computation times (Figure~\ref{fig:box}), where SmLAD requires substantially more time than PenLAD due to these larger matrix operations.

\begin{figure}[H]
\centering
\includegraphics[width = 0.6\textwidth]{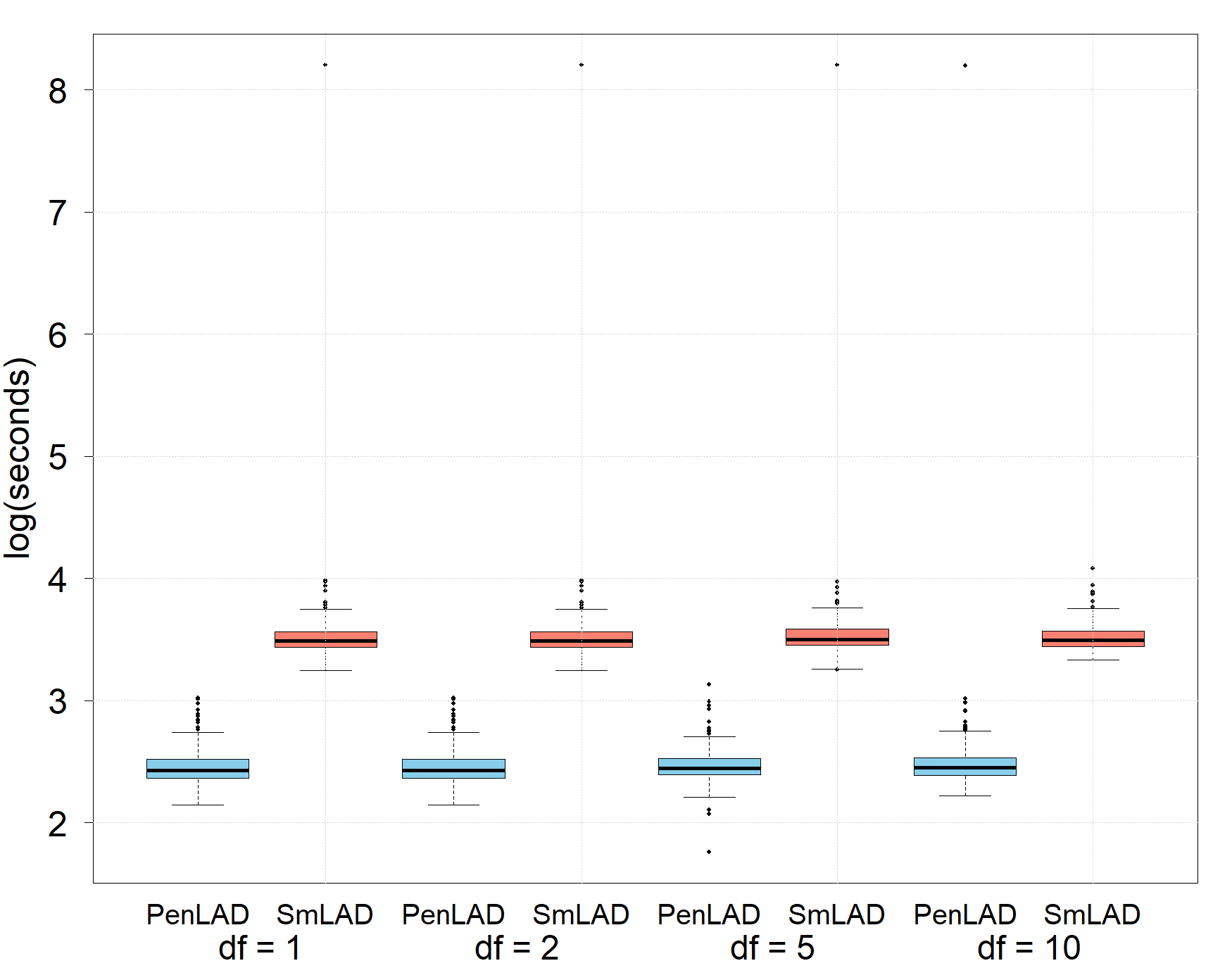} \
\caption{Log of computation times (in seconds) of PenLAD and SmLAD from 500 replications on an Intel Core i5-1600 computer.}
\label{fig:box}
\end{figure}

\section{Conclusion}
\label{sec:5}

The present paper provides theoretical and practical motivation for a family of lower-rank penalized spline estimators with flexible knots and penalties for functional location estimation. Under mild assumptions, these estimators are theoretically optimal, even attaining the parametric rate whenever the data are densely - though not necessarily completely -  observed. In practice, the proposed estimators match or exceed the robustness of existing methods while substantially reducing the computational burden relative to their main competitor. These properties make them well-suited for a wide variety of functional datasets, be it sparsely or densely sampled.

This work also opens up several promising directions for future research. Of particular interest is functional principal component analysis from discretely sampled functional data. Since the estimation problem can be formulated in a form analogous to \eqref{eq:est}, we expect that the methodology and insights developed herein will extend naturally to that setting. We plan to pursue this direction in forthcoming work.

\section{Appendix: Technical Lemmas}
\label{sec:app}

To lighten the notation in the course of our proofs we use $c_0$ to denote generic positive constants whose precise value is unimportant. Thus, the value of $c_0$ may change from appearance to appearance. To simplify our proofs we will, without loss of generality, identify $q$ in assumption \ref{A4} with the uniform density on $[0,1]$, i.e., $q = 1$ on $[0,1]$ and $q = 0$ otherwise. The symbols $\langle \cdot, \cdot \rangle$ and $\|\cdot\|$ will denote the standard $\mathcal{L}^2([0,1])$ inner product and associated norm while $\|\cdot\|_{\infty}$ denotes the sup-norm. Furthermore, we will occasionally use inequalities of the form $A \leq B$ for $A, B$ random variables defined on the same probability space $(\Omega,\mathcal{A},\mathbb{P})$. These inequalities are to be interpreted in the almost sure sense, i.e., $A(\omega) \leq B(\omega)$ for $\mathbb{P}$-almost all $\omega \in \Omega$.

The following two lemmas have been previously proven and we state them here for completeness. The first lemma concerns spline approximation of smooth functions whereas the second presents an RKHS-type inequality bounding the $\|\cdot\|_{\infty}$ norm on $S_{K}^p([0,1])$ in terms of the $\|\cdot\|_{r,\lambda}$ norm

\begin{lemma}[Theorem XII (6) \citep{DB:2000}]. For $j=0, \ldots, p-1$, there exists a $c_0 = c_0(j,p)$ so that, for all $t_1, \ldots, t_{K+p}$ with
\begin{align*}
t_1 = \ldots = t_p = 0 < t_{p+1} < \ldots < t_{K+p} < 1
\end{align*}
and for all $g \in \mathcal{C}^{(j)}([0,1])$,
\begin{align*}
\dist(g, S_{K}^p([0,1]) \leq c_0 h^{j} \sup_{|x-y| \leq h} |g^{(j)}(x) - g^{(j)}(y)|,
\end{align*}
where $\dist(g, S_{K}^p([0,1]) = \inf_{f \in S_{K}^p([0,1]}\|g-f\|_{\infty}$ and $h = \max_i (t_i-t_{i-1})$.
\label{Lem:SA}
\end{lemma}
\noindent
In particular, under assumption \ref{A2}, Lemma~\ref{Lem:SA} implies the existence of an $s_{\mu_0} \in S_{K}^p([0,1])$ with the property that
\begin{align}
\label{eq:SA}
\|\mu_0 - s_{\mu_0}\|_{\infty} = O(K^{-j}).
\end{align} 

\begin{lemma}[Proposition 1 of \citep{Kal:2023}] If Assumption \ref{A2} is satisfied and $p >r \geq 1$ then there exists $c_0>0$ such that for every $f \in S_{K}^p([0,1])$ and $\lambda \in [0,1]$, 
\begin{align*}
\sup_{t \in [0,1]}|f(t)| \leq c_0 \min(K^{1/2}, \lambda^{-1/(4r)}) \|f\|_{r,\lambda}. 
\end{align*}
\label{Lem:RKHS}
\end{lemma}
\noindent
Recall that equipped with $\langle \cdot, \cdot \rangle_{r,\lambda}$ is an RKHS with a symmetric reproducing kernel $\mathcal{R}_{p,r,\lambda}: [0,1]^2 \to \mathbbm{R}$ such that $x \mapsto \mathcal{R}_{p,r,\lambda}(x,y) \in S_{K}^p([0,1])$ for every $y \in [0,1]$. Therefore, by the reproducing property and Lemma \ref{Lem:RKHS},
\begin{align*}
 \| \mathcal{R}_{p,r,\lambda}(x,\cdot) \|_{r,\lambda}^2 =  \langle \mathcal{R}_{p,r,\lambda}(x,\cdot), \mathcal{R}_{p,r,\lambda}(\cdot,x) \rangle_{r,\lambda} = \mathcal{R}_{p,r,\lambda}(x,x)   \leq c_0 \min(K^{1/2}, \lambda^{-1/(4r)}) \| \mathcal{R}_{p,r,\lambda}(x,\cdot) \|_{r,\lambda}.
\end{align*}
If $\| \mathcal{R}_{p,r,\lambda}(x,\cdot) \|_{r,\lambda}>0$, dividing both sides with $\| \mathcal{R}_{p,r,\lambda}(x,\cdot) \|_{r,\lambda}$ yields
\begin{align*}
\| \mathcal{R}_{p,r,\lambda}(x,\cdot) \|_{r,\lambda}   \leq c_0 \min(K^{1/2}, \lambda^{-1/(4r)}),
\end{align*}
and since the right side is always positive, the inequality is also true whenever $\| \mathcal{R}_{p,r,\lambda}(x,\cdot) \|_{r,\lambda}=0$. Furthermore, as $c_0$ does not depend on $x$, the inequality is uniform for $x \in [0,1]$. In conclusion,
\begin{align}
\label{eq:rkineq}
\sup_{x \in [0,1]}\| \mathcal{R}_{p,r,\lambda}(x,\cdot) \|_{r,\lambda}   \leq c_0 \min(K^{1/2}, \lambda^{-1/(4r)}).
\end{align}

\begin{lemma}
\label{Lem:I1}
Under the conditions of Theorem~\ref{Thm:1},
\begin{align*}
\inf_{f \in S_{K}^p([0,1]):\|f\|_{r,\lambda} = D} I_1(f) & \geq c_0 D^2 C_n+ O(1) D C_n,
\end{align*}
as $n \to \infty$, i.e., \eqref{asser:1} holds.
\end{lemma}

\begin{proof}

We begin by observing that the domain of integration shrinks to zero as on the one hand, by Lemma~\ref{Lem:SA},
\begin{align*}
\max_{i,j}|R_{ij}| \leq \sup_{x \in [0,1]}|\mu_0(x) - s_{\mu_0}(x)| = O(K^{-j}) = o(1),
\end{align*}
since $K \to \infty$, and, on the other hand, by Lemma~\ref{Lem:RKHS}, for every $f \in S_{K}^p([0,1])$ such that $\|f\|_{r,\lambda} \leq D$ we have
\begin{align*}
C_n^{1/2}\sup_{x \in [0,1]} |f(x)| \leq c_0 C_n^{1/2} \min(K^{1/2},\lambda^{-1/(4r)}) D \leq  c_0 C_n^{1/2} K^{1/2} D = o(1),
\end{align*}
as $n \to \infty$, by our limit assumptions. Therefore, using assumption \ref{A7} we obtain
\begin{align*}
 \int_{R_{ij}}^{R_{ij}+C_n^{1/2}f(T_{ij})}\mathbb{E}[\psi(\epsilon_{ij}+u) - \psi(\epsilon_{ij}) \vert \mathcal{T}] \dd u & =  \int_{R_{ij}}^{R_{ij}+C_n^{1/2}f(T_{ij})} \{\delta(T_{ij})u + o(u)\} \dd u
 \\ &  = \int_{R_{ij}}^{R_{ij}+C_n^{1/2}f(T_{ij})}  \delta(T_{ij})u(1+o(1)) \dd u, \quad (\inf_{x \in [0,1]} \delta(x)>0).
 \\ & = c_0 \delta(T_{ij})C_n |f(T_{ij})|^2(1+o(1)) \\
 & \quad + c_0 C_n^{1/2}f(T_{ij})R_{ij}(1+o(1))
 \\ &  \geq c_0[\inf_{x \in [0,1]}  \delta(x)] C_n |f(T_{ij})|^2 - c_0 C_n^{1/2} |f(T_{ij})||R_{ij}|
 \\  & \geq  c_0 C_n |f(T_{ij})|^2 - c_0 C_n |f(T_{ij})|,  \quad (|R_{ij}| = O(K^{-j}) = O(C_n^{1/2})),
\end{align*}
for some $c_0>0$. Consequently, taking expectations with respect to $\mathcal{T}$, using assumption \ref{A4} and $\|f\|_{r,\lambda}^2 =  \|f\|^2 + \lambda\|f^{(r)}\|^2$, we find
\begin{align*}
I_1(f) & \geq c_0 C_n \|f\|^2 - c_0 C_n \mathbb{E}[|f(T_{ij})|] + C_n \lambda \|f^{(r)}\|^2 
\\ &\geq  c_0 C_n \|f\|^2 - c_0 C_n \{\mathbb{E}[|f(T_{ij})|^2]\}^{1/2} + \lambda C_n  \|f^{(r)}\|^2, \quad (\mathbb{E}[|X|] \leq \{\mathbb{E}[|X|^2]\}^{1/2})
\\ & \geq \min(c_0,1) C_n \|f\|_{r,\lambda}^2 - c_0 C_n \|f\|.
\end{align*}
So, taking the infimum and using $\|f\|\leq \|f\|_{r,\lambda}$, we finally obtain
\begin{align*}
\inf_{f \in S_{K}^p([0,1]):\|f\|_{r,\lambda} = D} I_1(f) & \geq c_0 D^2 C_n+ O(1) D C_n,
\end{align*}
which is the result of the Lemma.

\end{proof}

\begin{lemma}
\label{Lem:I2}
Under the conditions of Theorem~\ref{Thm:1},
\begin{align*}
\sup_{f \in S_{K}^p([0,1]):\|f\|_{r,\lambda} \leq D} |I_3(f)| &= O_{\mathbb{P}}(1) D C_n,
\end{align*}
as $n \to \infty$, i.e., \eqref{asser:3} holds.
\end{lemma}

\begin{proof}
For every $f \in S_{K}^p([0,1])$, by the reproducing property and the linearity of the inner product, we obtain
\begin{align*}
\frac{1}{n} \sum_{i=1}^n \frac{1}{m_i} \sum_{j=1}^{m_i} \psi(\epsilon_{ij}) f(T_{ij}) & = \frac{1}{n} \sum_{i=1}^n \frac{1}{m_i} \sum_{j=1}^{m_i} \psi(\epsilon_{ij}) \langle f, \mathcal{R}_{p,r,\lambda}(T_{ij},\cdot) \rangle_{r,\lambda}
\\ &  = \bigg \langle f, \frac{1}{n} \sum_{i=1}^n \frac{1}{m_i} \sum_{j=1}^{m_i} \psi(\epsilon_{ij})\mathcal{R}_{p,r,\lambda}(T_{ij},\cdot) \bigg \rangle_{r,\lambda}.
\end{align*}
Consequently, by the Schwarz inequality,
\begin{align*}
\sup_{f \in S_{K}^p([0,1]):\|f\|_{r,\lambda} \leq D} |I_3(f)| \leq C_n^{1/2} \left\| \frac{1}{n} \sum_{i=1}^n \frac{1}{m_i} \sum_{j=1}^{m_i} \psi(\epsilon_{ij})\mathcal{R}_{p,r,\lambda}(T_{ij},\cdot) \right\|_{r,\lambda} D.
\end{align*}
We now determine the order of the right side. Let $\mathcal{T}$ denote all sampling points. Squaring and taking conditional expectations, using assumptions \ref{A3} and \ref{A7} we obtain,
\begin{align*}
\mathbb{E} \left[ \left\| \frac{1}{n} \sum_{i=1}^n \frac{1}{m_i} \sum_{j=1}^{m_i} \psi(\epsilon_{ij})\mathcal{R}_{p,r,\lambda}(T_{ij},\cdot) \right\|_{r,\lambda}^2 \bigg  \vert \mathcal{T} \right] & = \frac{1}{n^2} \sum_{i=1}^n \frac{1}{m_i^2} \sum_{j=1}^{m_i} \sum_{k=1}^{m_i}  \mathcal{R}_{p,r,\lambda}(T_{ij},T_{ik}) \mathbb{E} [ \psi(\epsilon_{ij}) \psi(\epsilon_{ik}) \vert \mathcal{T}]
\\ & = Z_1 + Z_2,
\end{align*}
say, with 
\begin{align*}
Z_1 & = \frac{1}{n^2} \sum_{i=1}^n \frac{1}{m_i^2} \sum_{j=1}^{m_i}  \mathcal{R}_{p,r,\lambda}(T_{ij}, T_{ij}) \mathbb{E}[| \psi(\epsilon_{ij})|^2 \vert  \mathcal{T}] \\
Z_2 & = \frac{1}{n^2} \sum_{i=1}^n \frac{1}{m_i^2} \sum_{j \neq k } \mathcal{R}_{p,r,\lambda}(T_{ij}, T_{ik}) \mathbb{E}[ \psi(\epsilon_{ij}) \psi(\epsilon_{ik}) \vert \mathcal{T} ].
\end{align*}

We first bound $Z_1$ and subsequently $Z_2$. 	Taking the squared expectation out of the sum, we have
\begin{align*}
Z_1 \leq \sup_{t \in [0,1]} \mathbb{E}[| \psi(\epsilon_{1}(t))|^2]  \frac{1}{n^2} \sum_{i=1}^n \frac{1}{m_i^2} \sum_{j=1}^{m_i}  \mathcal{R}_{p,r,\lambda}(T_{ij}, T_{ij})  = O\left( \frac{\min(K, \lambda^{-1/(2r)})}{nm} \right),
\end{align*}
where we have used assumption \ref{A7}, \eqref{eq:rkineq} (after writing $\mathcal{R}_{p,r,\lambda}(T_{ij}, T_{ij}) = \|\mathcal{R}_{p,r,\lambda}(T_{ij},\cdot)\|_{r,\lambda}^2$) and $m^{-1} = n^{-1} \sum_{i=1}^n m_i^{-1}$. Note that this is a deterministic bound.

For $Z_2$, iterating expectations and using assumptions \ref{A3} and \ref{A4}, we obtain
\begin{align*}
\mathbb{E}[Z_2] & = \frac{1}{n^2} \sum_{i=1}^n \frac{m_i(m_i-1)}{m_i^2} \int_{0}^1 \int_{0}^1  \mathcal{R}_{p,r,\lambda}(x, y) \mathbb{E}\left[ \psi(\epsilon_1(x)) \psi(\epsilon_1(y)) \right] \dd x \dd y
\\ & = \frac{1}{n^2} \sum_{i=1}^n \frac{m_i(m_i-1)}{m_i^2}  \mathbb{E}\left[ \int_{0}^1 \int_{0}^1  \mathcal{R}_{p,r,\lambda}(x, y) \psi(\epsilon_1(x)) \psi(\epsilon_1(y))  \dd x \dd y \right], &&  (\text{Fubini's theorem})
\\ & \leq \frac{1}{n^2} \sum_{i=1}^n \frac{m_i(m_i-1)}{m_i^2}  \mathbb{E}\left[ \|\psi(\epsilon_1)\|^2\right], && (\text{Lemma} \ \ref{Lem:RK} \ \text{Part}\ \ref{itm:rk3})
\\ & = O(n^{-1}).
\end{align*}

Combining the bounds on $Z_1$ and $Z_2$, using the law of iterated expectations, we find
\begin{align*}
\mathbb{E} \left[ \left\| \frac{1}{n} \sum_{i=1}^n \frac{1}{m_i} \sum_{j=1}^{m_i} \psi(\epsilon_{ij})\mathcal{R}_{p,r,\lambda}(T_{ij},\cdot) \right\|_{r,\lambda}^2  \right] & = O\left(\frac{1}{n} +\frac{\min(K, \lambda^{-1/(2r)})}{nm}  \right).
\end{align*}
Markov's inequality yields 
\begin{align*}
\sup_{f \in S_{K}^p([0,1]):\|f\|_{r,\lambda} \leq D} |I_3(f)| \leq C_n^{1/2} O_{\mathbb{P}}\left( \left\{\frac{1}{n} +  \frac{\min(K, \lambda^{-1/(2r)})}{nm} \right\}^{1/2}\right) D = O_{\mathbb{P}}(1) D C_n,
\end{align*}
as the term inside the curly brackets is bounded above by $C_n^{1/2}$. This completes the proof.
\end{proof}

\begin{lemma}
\label{Lem:I4}

Under the conditions of Theorem~\ref{Thm:1},
\begin{align*}
\sup_{f \in S_{K}^p([0,1]): \|f\|_{r,\lambda} \leq D} |I_2(f)| = o_{\mathbb{P}}(1) C_n,
\end{align*}
as $n \to \infty$, i.e., \eqref{asser:2} holds.
\end{lemma}

\begin{proof}
Our proof is long and consists of three parts. On $S_{K}^p([0,1])$ and for $i=1, \ldots, n$, consider the independent mean-zero processes
\begin{align*}
Z_i(f) &= \frac{1}{m_i} \sum_{j=1}^{m_i} \left\{ \int_{R_{ij}}^{R_{ij}+C_n^{1/2}f(T_{ij})}  \{\psi\left(\epsilon_{ij}+u\right) 
- \psi(\epsilon_{ij}) \}\dd u \right.
\\ & \quad - 
\left. \mathbb{E} \left[\int_{R_{ij}}^{R_{ij}+C_n^{1/2}f(T_{ij})}  \{\psi\left(\epsilon_{ij}+u\right) - \psi(\epsilon_{ij}) \}\dd u \right]  \right\}, \quad (i=1, \ldots, n).
\end{align*}
Set $\mathcal{B}_D = \{f \in S_{K}^p([0,1]): \|f\|_{r,\lambda}\leq D\}$.  We will (i) derive a uniform bound of the $Z_i$ on $\mathcal{B}_D$, (ii) derive a uniform bound on the sum of the variances and (iii) derive a bound on the relevant bracketing number on $\mathcal{B}_D$. With all these ingredients we then apply Theorem 8.13 of \citep{VDG:2000}.

Let us begin by recalling that in the proof of Lemma~\ref{Lem:I1} we showed that the domain of integration shrinks to zero. So, applying the triangle inequality and assumption \ref{A5}, we get
\begin{align}
\label{eq:unbound}
\sup_{g \in \mathcal{B}_D} |Z_i(f)| \leq 2 M_1 C_n^{1/2} \sup_{f \in \mathcal{B}_D}\|f\|_{\infty} \leq c_0 C_n^{1/2} \min(K^{1/2}, \lambda^{-1/(4r)}) \leq c_0 C_n^{1/2} K^{1/2},
\end{align}
where to obtain the second-to-last inequality we have used Lemma \ref{Lem:RKHS} and we have absorbed $D$ into $c_0$ (the value of $D$ is not important for this derivation).

On the other hand, using the inequality $\Var[X] \leq \mathbb{E}[|X|^2]$, for the mean-centered processes $Z_i(f)$ we find
\begin{align*}
\frac{1}{n} \sum_{i=1}^n \Var[Z_i(f)] & \leq \frac{1}{n} \sum_{i=1}^n \frac{1}{m_i^2} \mathbb{E} \left[ \left| \sum_{j=1}^{m_i} \int_{R_{ij}}^{R_{ij}+C_n^{1/2}f(T_{ij})}  \{\psi\left(\epsilon_{ij}+u\right) - \psi(\epsilon_{ij}) \}\dd u \right|^2  \right]
\\ &  \leq \frac{1}{n} \sum_{i=1}^n \frac{1}{m_i} \sum_{j=1}^{m_i} \mathbb{E} \left[ \left| \int_{R_{ij}}^{R_{ij}+C_n^{1/2}f(T_{ij})}  \{\psi\left(\epsilon_{ij}+u\right) - \psi(\epsilon_{ij}) \}\dd u \right|^2 \right],   && (\text{By Schwarz}).
\end{align*}
Now, conditionally on $\mathcal{T}$, using the integral form of the Schwarz inequality, we have
\begin{align*}
 \mathbb{E} & \left[ \left| \int_{R_{ij}}^{R_{ij}+C_n^{1/2}f(T_{ij})}  \{\psi\left(\epsilon_{ij}+u\right) - \psi(\epsilon_{ij}) \}\dd u \right|^2 \bigg \vert \mathcal{T} \right]
 \\ & \quad \leq C_n^{1/2} |f(T_{ij})|\left| \int_{R_{ij}}^{R_{ij}+C_n^{1/2}f(T_{ij})}  \mathbb{E} [|\psi\left(\epsilon_{ij}+u\right) - \psi(\epsilon_{ij})|^2  | \mathcal{T}]\dd u \right|
 \\ & \leq c_0 C_n^{1/2} |f(T_{ij})| \left| \int_{R_{ij}}^{R_{ij}+C_n^{1/2}f(T_{ij})} |u| \dd u \right| \quad (\text{By Assumption} \ \ref{A6})
 \\ & \leq c_0 C_n^{3/2} |f(T_{ij})|(1+|f(T_{ij})|^2)
  \\ & \leq c_0 C_n^{3/2} \min(K^{1/2}, \lambda^{-1/(4r)})(1+|f(T_{ij})|^2) \quad (\text{By Lemma} \ \ref{Lem:RKHS}),
\end{align*}
and to obtain the second-to-last inequality we have used the antiderivative $\int |u| \dd u = u|u|/2$ and $K^{-j} \leq C_n^{1/2}$. The constant $c_0$ does not depend on $f \in \mathcal{B}_D$ or on $\mathcal{T}$, hence taking expectations with respect to $\mathcal{T}$ and using that, for $ f\in \mathcal{B}_D$, $\mathbb{E}[|f(T_{ij})|^2] = \|f\|^2 \leq \|f\|_{r,\lambda}^2 \leq D^2$, we obtain
\begin{align}
\label{eq:varbound}
\sup_{f \in \mathcal{B}_D} \left[n^{-1} \sum_{i=1}^n \Var[Z_i(f)] \right]  \leq c_0 C_n^{3/2} \min(K^{1/2}, \lambda^{-1/(4r)}) \leq c_0 C_n^{3/2} K^{1/2}.
\end{align}

Finally, take any $f,f^{\prime} \in S_{K}^p([0,1])$.  Recalling that the domain of integration shrinks to zero, assumption \ref{A4} yields
\begin{align}
\label{eq:Lipbound}
\max_{1 \leq i \leq n}  |Z_{i}(f) - Z_{i}(f^{\prime})| &\leq c_0 C_n^{1/2} \sup_{t \in [0,1]}|f(t)-f^{\prime}(t)| \nonumber
\\ &\leq c_0 C_n^{1/2} K^{1/2}  \|f-f^{\prime}\|,
\end{align}
where the bound of the sup-norm in terms of the $\mathcal{L}^2([0,1])$ follows as in the proof of Proposition 1 in \citep{Kal:2023}, which, in particular, shows that, under assumption \ref{A2}, $\|g\|_{\infty} \leq c_0 K^{1/2} \|g\|$ for every $g \in S_{K}^p([0,1])$. For every $f \in \mathcal{B}_D$, we have $\|f\|\leq \|f\|_{r,\lambda} \leq D$. By Corollary 2.6 of \citep{VDG:2000}, for the $\delta$-entropy in the $\mathcal{L}^2([0,1])$-norm of the class of functions $\mathcal{B}_D$, $H(\delta, \mathcal{B}_D, \|\cdot\|)$, we have the bound
\begin{align}
\label{eq:entrbound}
\mathcal{H}(\delta, \mathcal{B}_D, \|\cdot\|) \leq (K+p) \log \left( \frac{4D}{\delta}+ 1\right), \quad \delta>0.
\end{align}
But \eqref{eq:Lipbound} shows that a cover of $\mathcal{B}_D$ in the $\mathcal{L}^2([0,1])$ of radius $\delta /(C_n^{1/2} K^{1/2})$ provides a $\delta$-cover for the generalized entropy with bracketing of Definition 8.1 of \citep{VDG:2000}. Therefore, setting $\delta = u/(C_n^{1/2} K^{1/2})$, taking square roots of \eqref{eq:entrbound}, integrating from $0$ to $C_n^{3/4} K^{1/4}$ (the square root of \eqref{eq:varbound} up to a constant) and changing variables yields
\begin{align*}
(K+p)^{1/2}\int_{0}^{C_n^{3/4} K^{1/4}} \log^{1/2} \left( \frac{K^{1/2} C_n^{1/2}}{u}+1 \right) \dd u & = (K+p)^{1/2} C_n^{1/2}
\\ & \quad \times  \int_{0}^{C_n^{1/4}K^{1/4}} \log^{1/2}\left(\frac{K^{1/2}}{u}+1 \right) \dd u
\\ & \sim K^{3/4} C_n^{3/4} \log^{1/2}(n),
\end{align*}
as, by assumption, $n^{\delta-1}K^3 \to 0$ for some $\delta>0$ and so $K \leq c_0 n^{(1-\delta)/3}$. Taking $\alpha := K^{3/4} C_n^{3/4} \log^{1/2}(n)$ in the notation of \citep{VDG:2000}, the assumptions (8.43) and (8.45) are satisfied. Therefore, Theorem 8.13 of that author applies yielding
\begin{align*}
\sup_{g \in \mathcal{B}_D} \left|n^{-1/2} \sum_{i=1}^n Z_i(f) \right| =  O_{\mathbb{P}}(K^{3/4} C_n^{3/4} \log^{1/2}(n) ).
\end{align*}
But $I_2(f) = n^{-1} \sum_{i=1}^n Z_i(f)$, so using that $C_n^{-1/4} \leq n^{1/4}$, we find 
\begin{align*}
\frac{\sup_{f \in \mathcal{B}_D} |I_2(f)|}{C_n} = O_{\mathbb{P}}(1) \frac{K^{3/4} \log^{1/2}(n)}{n^{1/2} C_n^{1/4}} \leq O_{\mathbb{P}}(1) \frac{K^{3/4} \log^{1/2}(n)}{n^{1/4}} \to 0,
\end{align*}
in probability, by our limit assumption $n^{\delta-1} K^3 \to 0 $ for some small $\delta>0$. This concludes the proof of the lemma.

\end{proof}

\end{document}